\documentclass[acmsmall,screen,nonacm]{acmart}
\AtBeginDocument{%
  }

\usepackage{amsmath,amsfonts}
\usepackage{algorithmic}
\usepackage{adjustbox}
\usepackage{graphicx}
\usepackage{textcomp}
\usepackage{xcolor}
\usepackage{physics}
\usepackage{algorithm}
\usepackage{wrapfig}
\usepackage{subcaption}
\usepackage{diagbox}
\usepackage{multirow}
\usepackage{multicol}
\usepackage{makecell}
\usepackage{array}
\usepackage{enumitem}
\usepackage{hyperref}
\usepackage{bm}
\hypersetup{}

\newtheorem{theorem}{Theorem}[section]

\newtheorem{definition}{Definition}[section]
\newtheorem{proposition}{Proposition}[section]
\newtheorem{example}{Example}[section]
\newtheorem{principle}{Principle}

\newcommand{\myCompilerName}{TopoLS} 
\newcommand{\myCompilerNameSpace}{TopoLS }

\makeatletter\@secpenalty=0\makeatother  % with raggedbottom, do not prefer page breaks before section headings (they left pages short)
\begin{document}

%%
%% The "title" command has an optional parameter,
%% allowing the author to define a "short title" to be used in page headers.
\title{TopoLS: Lattice Surgery Compilation via Topological Program Transformations}

%%
%% The "author" command and its associated commands are used to define
%% the authors and their affiliations.
%% Of note is the shared affiliation of the first two authors, and the
%% "authornote" and "authornotemark" commands
%% used to denote shared contribution to the research.
\author{Junyu Zhou}
\affiliation{%
  \institution{University of Pennsylvania}
  \city{Philadelphia}
  \country{USA}}
\affiliation{%
  \institution{Lawrence Berkeley National Laboratory}
  \city{Berkeley}
  \country{USA}}
\email{junyuzh@seas.upenn.edu}

\author{Yuhao Liu}
\affiliation{%
  \institution{University of Pennsylvania}
  \city{Philadelphia}
  \country{USA}}
\email{liuyuhao@seas.upenn.edu}

\author{Ethan Decker}
\affiliation{%
  \institution{University of Pennsylvania}
  \city{Philadelphia}
  \country{USA}}
\email{ecd5249@upenn.edu}

\author{Justin Kalloor}
\affiliation{%
  \institution{University of California, Berkeley}
  \city{Berkeley}
  \country{USA}}
\email{jkalloor3@berkeley.edu}

\author{Mathias Weiden}
\affiliation{%
  \institution{University of California, Berkeley}
  \city{Berkeley}
  \country{USA}}
\email{mtweiden@berkeley.edu}

\author{Kean Chen}
\affiliation{%
  \institution{University of Pennsylvania}
  \city{Philadelphia}
  \country{USA}}
\email{keanchen@seas.upenn.edu}

\author{Costin Iancu}
\affiliation{%
  \institution{Lawrence Berkeley National Laboratory}
  \city{Berkeley}
  \country{USA}}
\email{ciancu@lbl.gov}

\author{Gushu Li}
\affiliation{%
  \institution{University of Pennsylvania}
  \city{Philadelphia}
  \country{USA}}
\email{gushuli@seas.upenn.edu}

%%
%% By default, the full list of authors will be used in the page
%% headers. Often, this list is too long, and will overlap
%% other information printed in the page headers. This command allows
%% the author to define a more concise list
%% of authors' names for this purpose.
\renewcommand{\shortauthors}{Zhou et al.}

%%
%% The abstract is a short summary of the work to be presented in the
%% article.

\begin{abstract}
 % Fault-tolerant quantum computing with surface codes is typically realized by compiling logical circuits into lattice surgery programs. However, existing compilers are circuit-centric and operate directly on gate sequences, limiting their ability to exploit the topological flexibility of lattice surgery layouts. We present \myCompilerName, a topology-centric compiler that uses ZX diagrams as an intermediate representation for automatic lattice surgery compilation. \myCompilerNameSpace combines ZX-diagram transformations with Monte Carlo Tree Search (MCTS) to construct feasible pipe-diagram embeddings, together with topology-aware circuit partitioning to ensure scalability. Evaluations on benchmark algorithms across multiple architectures show that \myCompilerNameSpace reduces space--time volume by an average of $46\%$ over prior circuit-centric compilers while showing near-linear compilation-time growth. Compared to SAT-based approaches that become intractable for larger circuits, \myCompilerNameSpace provides a scalable solution for lattice surgery compilation.

  Lattice surgery is a leading approach for implementing fault-tolerant logical operations in surface code quantum computing. Current compilers largely treat lattice surgery as a layout problem downstream of a gate-level circuit; this circuit-centric view hinders a substantial optimization opportunity: lattice surgery computations are topological programs whose semantics are determined by connectivity, so different concrete space--time realizations may implement the same computation while offering different space--time volumes. Circuit-centric compilers, however, commit to fixed realizations throughout the program. We tackle this problem by presenting \myCompilerName, a topology-centric compiler that uses ZX diagrams as the intermediate representation (IR) for automated lattice surgery compilation. Compiler transformations on this graph preserve its semantics, and the backend must produce a valid pipe diagram, i.e., a 3D space--time layout of the lattice surgery program, with the same semantics. \myCompilerNameSpace first simplifies the graph and slices it into layers, and then uses Monte Carlo Tree Search to place the graph's nodes as cubes in a 3D grid and to route its edges as pipes. To scale to large circuits, \myCompilerNameSpace partitions the graph so that each search step routes a bounded number of connections. Across evaluated benchmarks, \myCompilerNameSpace reduces space--time volume by $46\%$ on average over prior circuit-centric compilers, with improvements ranging from $25\%$ to $90\%$, while exhibiting strong empirical scalability on large benchmark families.
\end{abstract}

%%
%% The code below is generated by the tool at http://dl.acm.org/ccs.cfm.
%% Please copy and paste the code instead of the example below.
%%

%%
%% Keywords. The author(s) should pick words that accurately describe
%% the work being presented. Separate the keywords with commas.
% \keywords{Do, Not, Use, This, Code, Put, the, Correct, Terms, for,
%   Your, Paper}

% \received{20 February 2007}
% \received[revised]{12 March 2009}
% \received[accepted]{5 June 2009}

%%
%% This command processes the author and affiliation and title
%% information and builds the first part of the formatted document.
\maketitle

% add the paper content here
\section{Introduction} \label{Sec:Introduction}

Quantum error correction (QEC) is essential for achieving scalable and reliable quantum computation~\cite{shor1995scheme, steane1996error, chow2014implementing}. 
Quantum hardware is inherently noisy, suffering from decoherence~\cite{schlosshauer2019quantum}, relaxation~\cite{carroll2022dynamics}, and measurement errors~\cite{funcke2022measurement}, which accumulate rapidly as circuits grow in size. 
Without QEC, even moderately deep computations become unreliable. 
Among the many proposed codes, the surface code~\cite{fowler2012surface, knill1998resilient, dennis2002topological} is widely regarded as a leading candidate due to its high error threshold, locality of interactions, and compatibility with two-dimensional hardware layouts. 
These properties make it particularly attractive for large-scale fault-tolerant quantum computing and for executing practical quantum algorithms~\cite{gidney2021factor}. 
Recent experimental demonstrations of surface code primitives further support its feasibility~\cite{google2024quantum, google2023suppressing, Krinner_2022}.

\begin{figure}[t]
    \centering
    \includegraphics[width=1\linewidth]{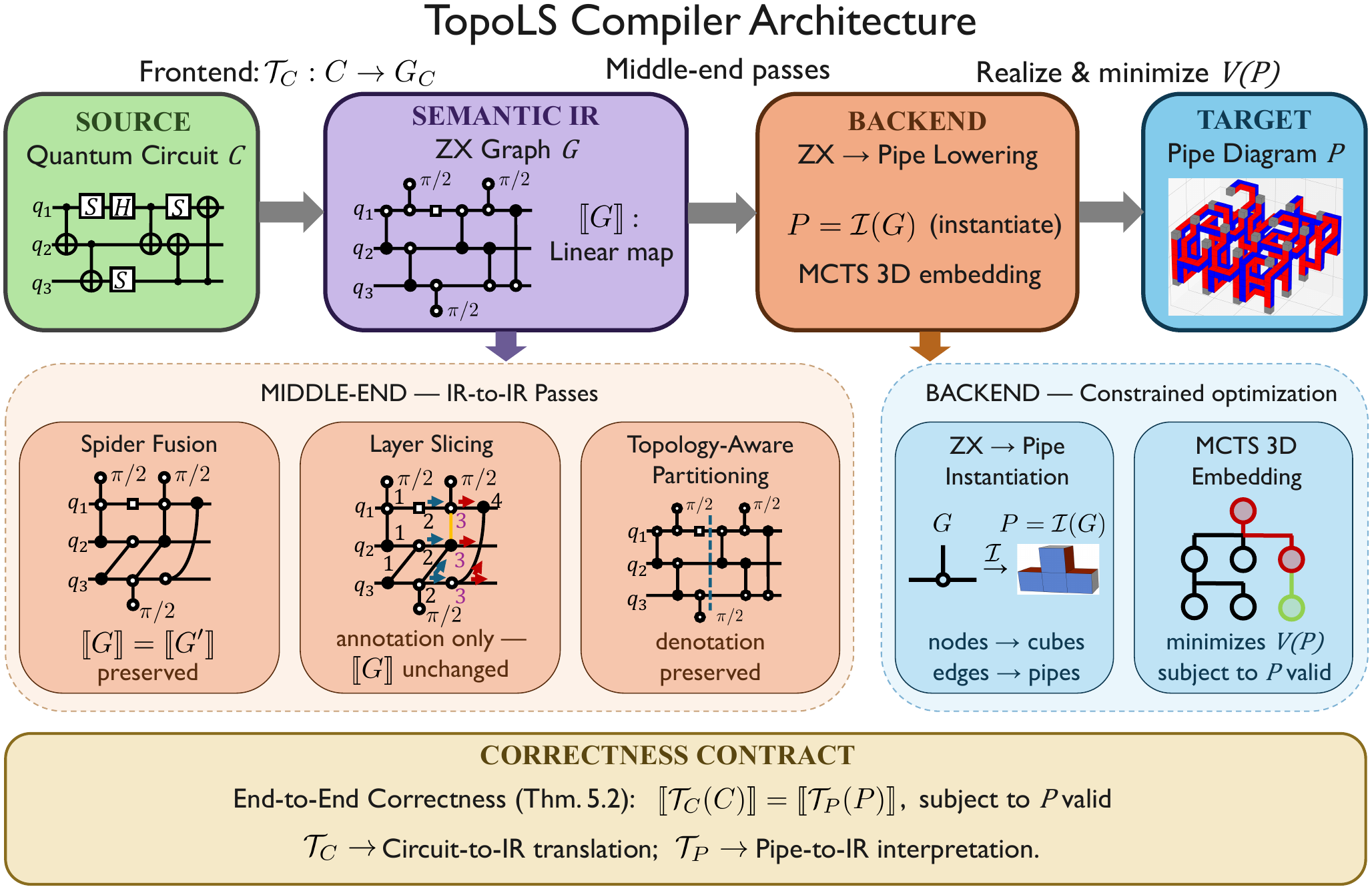}
    \caption{Overview of \myCompilerNameSpace compiler architecture}
    \label{fig:Overview}
\end{figure}

A major practical challenge in surface code-based quantum computing is not only how to implement logical operations, but how to represent them so that a compiler can optimize them. Realistic hardware platforms with nearest-neighbor connectivity, such as superconducting qubits~\cite{kjaergaard2020superconducting, krantz2019quantum} and spin qubits~\cite{garcia2021semiconductor, veldhorst2017silicon}, do not directly provide the nonlocal interactions assumed by high-level circuits. Lattice surgery~\cite{horsman2012surface, fowler2018low, fowler2012time, gidney2019efficient} addresses this mismatch by implementing logical operations through controlled merging and splitting of code patches. The resulting programs are naturally space--time objects: their cost is measured by \emph{space--time volume}~\cite{fowler2018low,herr2017lattice}, and their behavior is determined by topological connectivity rather than by a fixed geometric drawing. This makes lattice surgery compilation a programming-languages problem as much as a layout problem: the compiler needs an intermediate representation that exposes semantic equivalence, supports correctness-preserving transformations, and admits lowering to valid space--time realizations.

%Compiling lattice surgery programs is fundamentally challenging due to the interplay of three tightly coupled factors.

In particular, we identify three key challenges of compiling and optimizing lattice surgery programs on surface code hardware.

\textbf{First, Representation Challenge.}
Lattice surgery is inherently \emph{topological}: its semantics are determined by connectivity between QEC blocks rather than by a fixed gate sequence or a fixed geometric route. ZX-calculus~\cite{coecke2011interacting} provides exactly such a connectivity-based semantics and has been used to reason about lattice surgery constructions and verify correctness~\cite{de2020zx,gidney2024inplace,tan2024sat}. However, existing compilers remain largely circuit-centric~\cite{watkins2024high, molavi2025dependency}, while prior ZX-based uses are mostly analytical rather than compiler-IR driven. What is missing is a compiler that uses such a representation as its IR, with a denotation, well-formed target realizations, and correctness-preserving passes.

\textbf{Second, Realization Challenge.}
Given an optimized topological program, the compiler must synthesize a valid pipe diagram that realizes the same semantics. This is not merely drawing the graph: the target must satisfy a series of pipe-language constraints while choosing placements and pipe routes in 3D space--time. Exact formulations such as SAT-based encodings~\cite{tan2024sat} provide useful specifications but become difficult to scale as the realization space grows.

\textbf{Third, Scalability Challenge.}
The compiler must remain efficient for large programs. 
Naively exploring the full realization space leads to exponential growth, while overly restrictive heuristics sacrifice optimization quality. 
Achieving both high-quality realizations and scalable compilation remains an open challenge.

In this paper, we present \myCompilerName, a compiler architecture that treats ZX-calculus as a semantic IR for lattice surgery. The key design is to separate \emph{what} computation is represented from \emph{how} it is realized in space--time, allowing the compiler to optimize the ZX-level program before committing to concrete pipe routes.
The overview of \myCompilerNameSpace is shown in Figure~\ref{fig:Overview}. \textbf{First,} the frontend transforms a quantum circuit into the IR. We formalize the compiler IR as a boundary-indexed ZX graph (Section~\ref{Sec:TopoLSIR}). This IR records the program boundary, qubit labels, generator types, phases, and topological connectivity while abstracting away concrete space--time routes.
\textbf{Second,} as shown in the MIDDLE-END box, we define denotation-preserving passes over the ZX IR, including spider fusion and layer slicing (Section~\ref{Sec:TopologicalOptimizationat}). These passes simplify the graph and expose parallel merge--split opportunities before any physical routes are fixed. For scalability, the middle-end also uses topology-aware IR partitioning to bound the routing frontier exposed to later search (Section~\ref{Sec:Scalability}).
\textbf{Third,} the backend lowers the optimized ZX graph IR to an actual lattice-surgery pipe diagram. We formulate lowering as a constrained realization problem: the backend instantiates ZX nodes as cubes, constructs pipe routes for ZX edges, and uses MCTS-based 3D embedding to search for a valid low-volume pipe diagram (Section~\ref{Sec:3DLayoutOptimization}). Correctness is ensured by guaranteeing that the constructed valid pipe diagram, when interpreted as a ZX graph, has the same semantics as the ZX graph induced by the original quantum circuit.
Together, these components form an end-to-end compiler whose passes are organized around a semantic IR contract, enabling practical compilation of large lattice surgery programs with high-quality pipe-diagram realizations.

\textbf{The core problem in graph terms.} Abstracting away the quantum details, \myCompilerNameSpace solves a constrained graph-layout problem: it places every node of a typed graph on a distinct cell of a 3D grid (two spatial axes and time) and realizes every edge as a path of unit-length pipes, possibly through intermediate nodes that act as wires, such that paths do not overlap, every node has at most four pipes in one plane, and pipe colors match their endpoints (Definition~\ref{def:ValidPipeDiagram}); the objective is the volume of the bounding box. Before the layout, the compiler may rewrite the graph without changing its meaning, e.g., by merging adjacent nodes of the same type. Readers without a quantum-computing background can skip Sections~\ref{Sec:QubitGateStabilizer}--\ref{Sec:LatticeSurgeryforSurfaceCode}.

\textbf{What is inherited and what is new.} \myCompilerNameSpace builds on the ZX-calculus~\cite{coecke2011interacting}, its correspondence with lattice surgery~\cite{de2020zx}, and pipe diagrams~\cite{tan2024sat}. Prior work uses them to design, explain, or verify hand-made protocols such as magic-state factories~\cite{gidney2019efficient} and algorithm layouts~\cite{gidney2021factor}, and LaSsynth~\cite{tan2024sat} synthesizes volume-optimal pipe diagrams for small Clifford blocks with a SAT solver. \myCompilerNameSpace contributes an automated pipeline from circuits: ZX-level transformation, BFS slicing, topology-aware partitioning, and layer-by-layer MCTS realization. Compared with LaSsynth, it gives up optimality guarantees for scalability to larger and non-Clifford circuits.

%Across a range of benchmarks, \myCompilerNameSpace achieves an average $46\%$ reduction in space--time volume over state-of-the-art baselines~\cite{watkins2024high, molavi2025dependency}, with improvements ranging from $25\%$ to $90\%$. 
%These gains are consistent across hardware configurations, with $35\%$ to $38\%$ reductions observed on grids of varying shapes and sizes. 
%In addition, \myCompilerNameSpace demonstrates strong empirical scalability, compiling large circuits with near-linear runtime growth across evaluated benchmarks.

Across evaluated benchmarks, \myCompilerNameSpace achieves an average $46\%$ reduction in space--time volume over state-of-the-art baselines~\cite{watkins2024high, molavi2025dependency}, with improvements ranging from $25\%$ to $90\%$. 
These gains remain consistent across the hardware configurations we study, with $35\%$ to $38\%$ reductions observed on grids of varying shapes and sizes. 
In addition, \myCompilerNameSpace shows strong empirical scalability: on the evaluated benchmark families, compilation time exhibits near-linear growth with circuit size.

The major contributions of this paper can be summarized as follows:
\begin{enumerate}
\item We use ZX graphs as the semantic IR of an automated lattice-surgery compiler, so that the compiler can optimize the topological structure of the program across gate boundaries before it commits to a space--time layout. 
\item We develop denotation-preserving compiler passes over the ZX IR. These passes, including spider fusion and layer slicing, simplify the graph and expose merge--split parallelism.
\item We formulate backend lowering as a constrained IR realization problem and implement it with MCTS-based 3D pipe synthesis, together with topology-aware IR partitioning that bounds the routing frontier for scalability.
\item We demonstrate that \myCompilerNameSpace achieves substantial reductions in space--time volume (46\% on average), with consistent improvements across different hardware configurations and strong empirical scalability across benchmark algorithms.
\end{enumerate}

\section{Background} \label{Sec:Preliminary}

This section reviews the lattice surgery, ZX-calculus, and Monte Carlo Tree Search (MCTS) background used by our compiler. We refer readers to~\cite{nielsen2010quantum} for general background on quantum computing and to the cited subsection references for further detail.

\subsection{Physical Qubit, Quantum Operation, and Stabilizer}
\label{Sec:QubitGateStabilizer}

A physical qubit is a unit vector in \(\mathbb{C}^2\); in the computational basis, $|\psi\rangle = \alpha |0\rangle + \beta |1\rangle$ with $ |\alpha|^2 + |\beta|^2 = 1$. Here the kets $|0\rangle,|1\rangle$ form the computational ($Z$) basis, $|\pm\rangle=(|0\rangle\pm|1\rangle)/\sqrt{2}$ form the $X$ basis, and the bra $\langle\psi|$ is the conjugate transpose of $|\psi\rangle$. Operations on \(n\) physical qubits are unitary operators on \((\mathbb{C}^2)^{\otimes n}\). The single-qubit Pauli operators are the identity $I$, the bit flip $X=\left(\begin{smallmatrix}0&1\\1&0\end{smallmatrix}\right)$, the phase flip $Z=\left(\begin{smallmatrix}1&0\\0&-1\end{smallmatrix}\right)$, and $Y=iXZ$; an $n$-qubit Pauli operator is a tensor product of single-qubit Pauli operators times a phase $c\in\{\pm1,\pm i\}$.
Besides Pauli operators, input circuits use the Hadamard gate $H$, which maps $|0\rangle\mapsto|+\rangle$ and $|1\rangle\mapsto|-\rangle$; the phase gates $R_z(\theta)=|0\rangle\langle 0|+e^{i\theta}|1\rangle\langle 1|$ (up to a global phase), with $S=R_z(\pi/2)$ and $T=R_z(\pi/4)$; and the two-qubit CNOT gate, which maps $|c\rangle|t\rangle\mapsto|c\rangle|t\oplus c\rangle$ for a control qubit $c$ and a target qubit $t$.

Measurements project a quantum state onto observable eigenspaces. We mainly use Pauli measurements: measuring a Hermitian Pauli operator \(P\) returns \(s \in \{+1,-1\}\).
When \(P\) acts nontrivially on multiple qubits, the measurement extracts a joint parity without fully measuring each qubit; these multi-qubit Pauli measurements underlie both surface code syndrome extraction and lattice surgery merge--split operations.

The stabilizer formalism specifies quantum error-correcting codes as simultaneous \(+1\) eigenspaces of commuting Pauli operators. A Pauli operator \(S\) stabilizes \(|\psi\rangle\) if $S|\psi\rangle = |\psi\rangle$.
A stabilizer code is the common \(+1\) eigenspace of an Abelian group of Hermitian Pauli operators (Pauli operators with \(c=\pm 1\), which equal their own conjugate transpose) that does not contain \(-I\); its elements are the stabilizers. Measuring its generators (the checks) reveals errors without disturbing the logical state.

\subsection{Rotated Surface Code} \label{Sec:RotatedSurfaceCode}

\begin{wrapfigure}{R}{0.4\textwidth}
    \centering
    \includegraphics[width=0.4\textwidth]{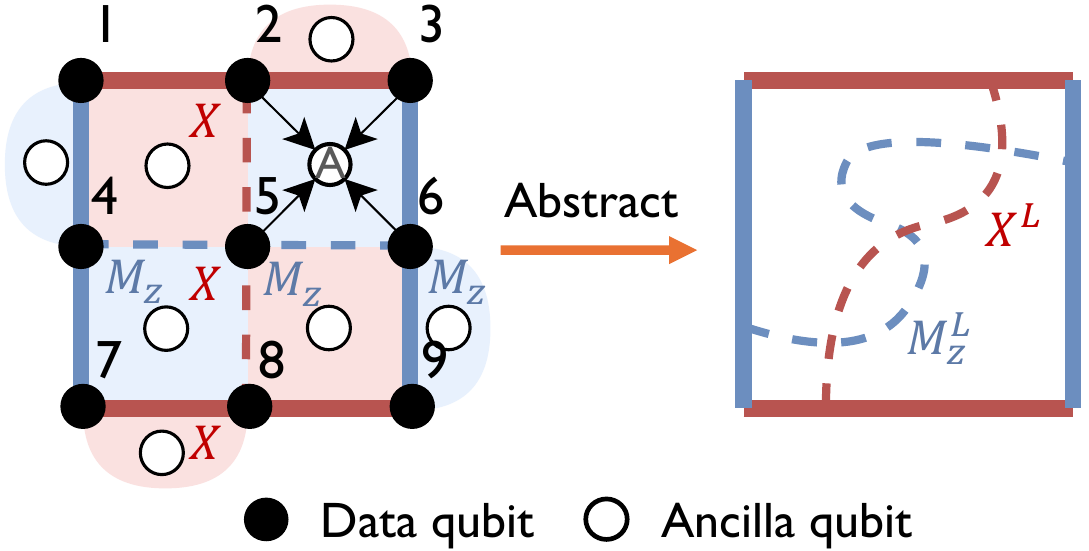}
    % \vspace{-5pt}
    \caption{Logical qubit and patch abstraction}
    % \vspace{-5pt}
    \label{fig:LogicalQubit}
\end{wrapfigure}

We focus on the 2D rotated surface code~\cite{fowler2012surface,horsman2012surface}, a leading fault-tolerant architecture because of its high threshold, 2D nearest-neighbor interactions, and efficient logical operations~\cite{google2024quantum}.

Figure~\ref{fig:LogicalQubit} shows a distance-3 rotated surface code patch encoding one logical qubit with nine data qubits (black dots) and ancilla qubits (white dots) for local parity checks. A \emph{patch} is the square region of physical qubits that encodes one logical qubit; a distance-$d$ patch has $d\times d$ data qubits, and a logical error requires errors on at least $d$ of them. Blue regions are multi-$Z$ checks, i.e., measurements of the product of $Z$ operators on the data qubits at the corners of the region, detecting bit-flip ($X$) errors; red regions are multi-$X$ checks detecting phase-flip ($Z$) errors. For example, ancilla A measures stabilizer $Z_2Z_3Z_6Z_5$ on its four neighboring data qubits.

\begin{wrapfigure}{R}{0.4\textwidth}
    \centering
    \includegraphics[width=0.4\textwidth]{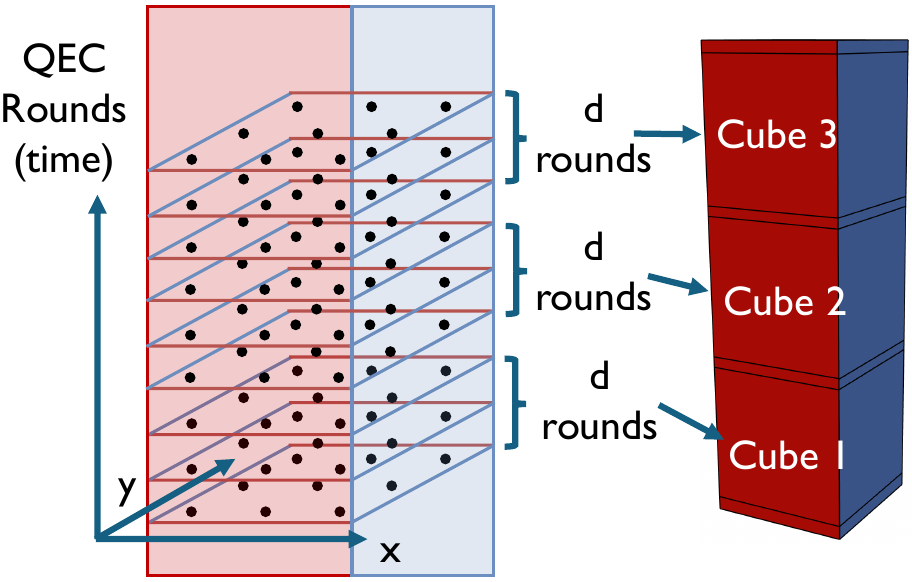}
    % \vspace{-5pt}
    \caption{QEC block representation}
    % \vspace{-5pt}
    \label{fig:QECBlock}
\end{wrapfigure}

These checks also define the colored boundaries of the patch. Logical Pauli operators can be viewed as paths connecting opposite boundaries of the appropriate type. For instance, a logical $X$ operator, denoted $X^L$, can be represented by applying physical $X$ operators along a vertical path such as qubits 2, 5, and 8. Similarly, a logical $Z$-basis measurement, denoted $M_z^L$, can be represented by measuring a horizontal path such as qubits 4, 5, and 6 in the $Z$ basis. In practice, logical Pauli operations such as $X^L$ and $Z^L$ are usually tracked in software using the Pauli frame~\cite{fowler2012surface}, a classical record of the Pauli corrections pending on each logical qubit, while other logical operations require additional procedures discussed later.

 For larger code distances, the patch contains more physical qubits, which improves error protection. However, the same high-level structure remains: colored boundaries determine the logical operators, and logical operations are represented by paths across the patch. Therefore, in the rest of the paper, we use the simplified patch abstraction shown on the right of Figure~\ref{fig:LogicalQubit}.

In addition to the spatial layout of a surface code patch, error correction also has a temporal dimension. Since the syndrome measurements themselves can be faulty, a single round of measurements is not sufficient to reliably distinguish data-qubit errors from measurement errors. Instead, the same stabilizer checks are repeated over multiple QEC rounds. For a distance-$d$ surface code patch, we typically group $d$ consecutive syndrome-measurement rounds into one fault-tolerant QEC block, as illustrated on the left of Figure~\ref{fig:QECBlock}. As before, we can abstract away the underlying physical-qubit layout and represent one $d$-round QEC block as a colored three-dimensional cube, whose horizontal axes ($x$ and $y$ in Figure~\ref{fig:QECBlock}) are spatial and whose vertical axis is time, as shown on the right of Figure~\ref{fig:QECBlock}.

\subsection{Lattice Surgery on Surface Code} \label{Sec:LatticeSurgeryforSurfaceCode}

% \begin{wrapfigure}{R}{0.17\textwidth}
%     \centering
%     % \vspace{-10pt}
%     \includegraphics[width=0.17\textwidth]{fig/MergeSplit.pdf}
%     % \vspace{-10pt}
%     \caption{Merge and Split}
%     % \vspace{-10pt}
%     \label{fig:MergeSplit}
% \end{wrapfigure}

\iffalse
%One efficient method for executing a quantum program using the surface code error correction is through 
Lattice surgery~\cite{horsman2012surface} can enable efficient multi-logical-qubit operations on the rotated surface code, where it can be reasoned as merges and splits.

An example is shown in Figure~\ref{fig:MergeSplit}(a). When the blue ($Z$) boundaries of two logical qubits are merged, the product of their logical $X^L$ operators becomes connected and is measured as $X_1^LX_2^L$. The subsequent split operation separates the two logical qubits after this measurement. Similarly, a $Z_1^LZ_2^L$ measurement can be performed by merging and then splitting the red ($X$) boundaries. Further details are provided in~\cite{horsman2012surface}.

\fi

\begin{wrapfigure}{R}{0.5\textwidth}
    \centering
    \includegraphics[width=0.5\textwidth]{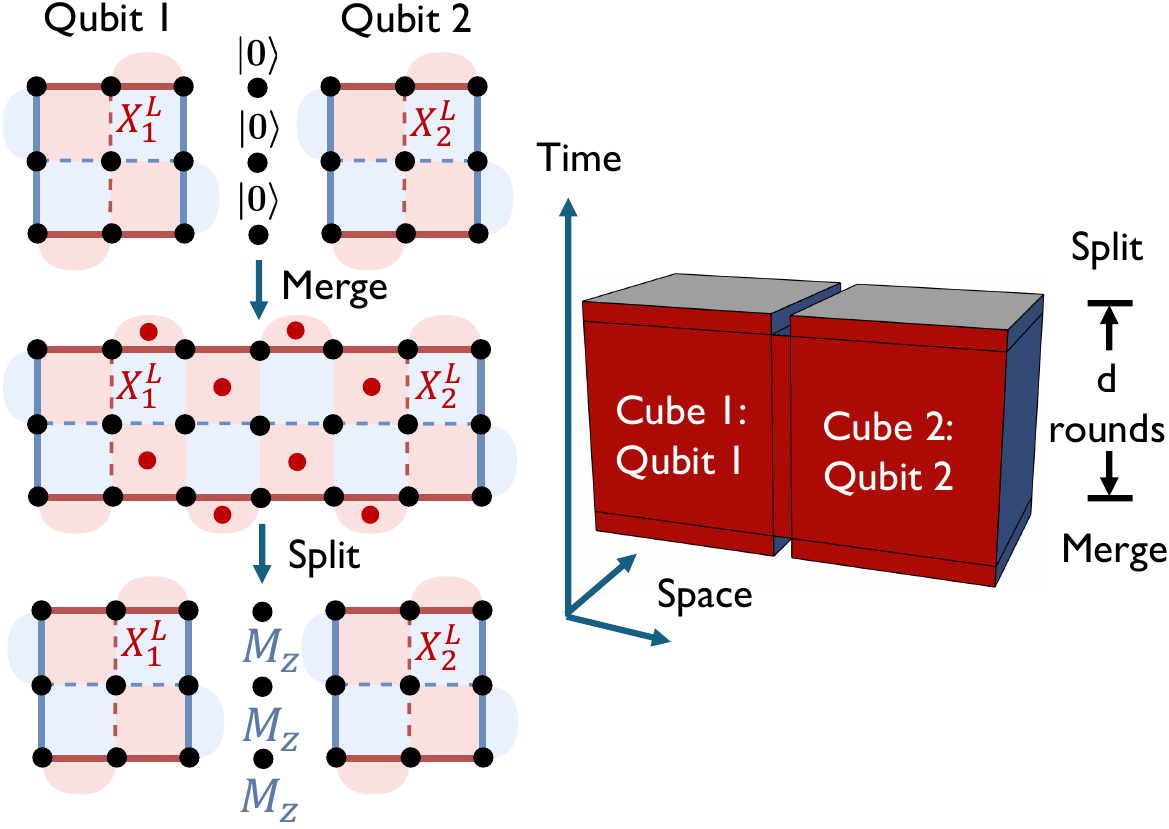}
    % \vspace{-5pt}
    \caption{Lattice surgery}
    % \vspace{-5pt}
    \label{fig:LS}
\end{wrapfigure}
% FIGURE TODO (A#10, optional): the right panel has two space arrows but one "Space" label;
%   label both (e.g., Space x, Space y), since the text now says "two space axes and one time axis".

Lattice surgery~\cite{horsman2012surface} implements multi-logical-qubit operations by temporarily deforming patches through merge and split operations. The left of Figure~\ref{fig:LS} measures $X_1^LX_2^L$: two logical qubits start in separate patches, with intermediate physical qubits initialized in $|0\rangle$. The merge activates these qubits and combines the patches, creating new $X$ stabilizer checks (red) and $Z$ stabilizer checks (blue). In the merged patch, the logical strings for $X_1^L$ and $X_2^L$ connect, and $X_1^LX_2^L$ is obtained as the product of the newly formed interface stabilizers highlighted by red dots. The split restores two logical qubits by measuring the intermediate qubits in the $Z$ basis. A $Z_1^LZ_2^L$ measurement similarly merges and splits the complementary boundaries~\cite{horsman2012surface}.

For a distance-$d$ patch, the merged stabilizer configuration is maintained and measured for $d$ QEC rounds before the split to preserve distance-$d$ fault tolerance. Using the QEC-block abstraction, we omit physical-qubit details and represent merge--split as a connection between colored 3D space-time blocks (Figure~\ref{fig:LS}, right), with two space axes and one time axis. The intermediate ancilla qubits occupy only a narrow interface between patches and are typically omitted from the high-level spatial footprint.

% \begin{figure}[t]
%     \centering
%     \includegraphics[width=0.6\linewidth]{fig/MergeSplit.pdf}
%     % \vspace{-5pt}
%     \caption{Lattice surgery and pipe diagram}
%     % \vspace{-5pt}
%     \label{fig:MergeSplit}
% \end{figure}

\subsection{Pipe Diagram for Lattice Surgery} \label{Sec:PipeDiagramforLatticeSurgery}

\begin{figure}[t]
    \centering
    \includegraphics[width=0.7\textwidth]{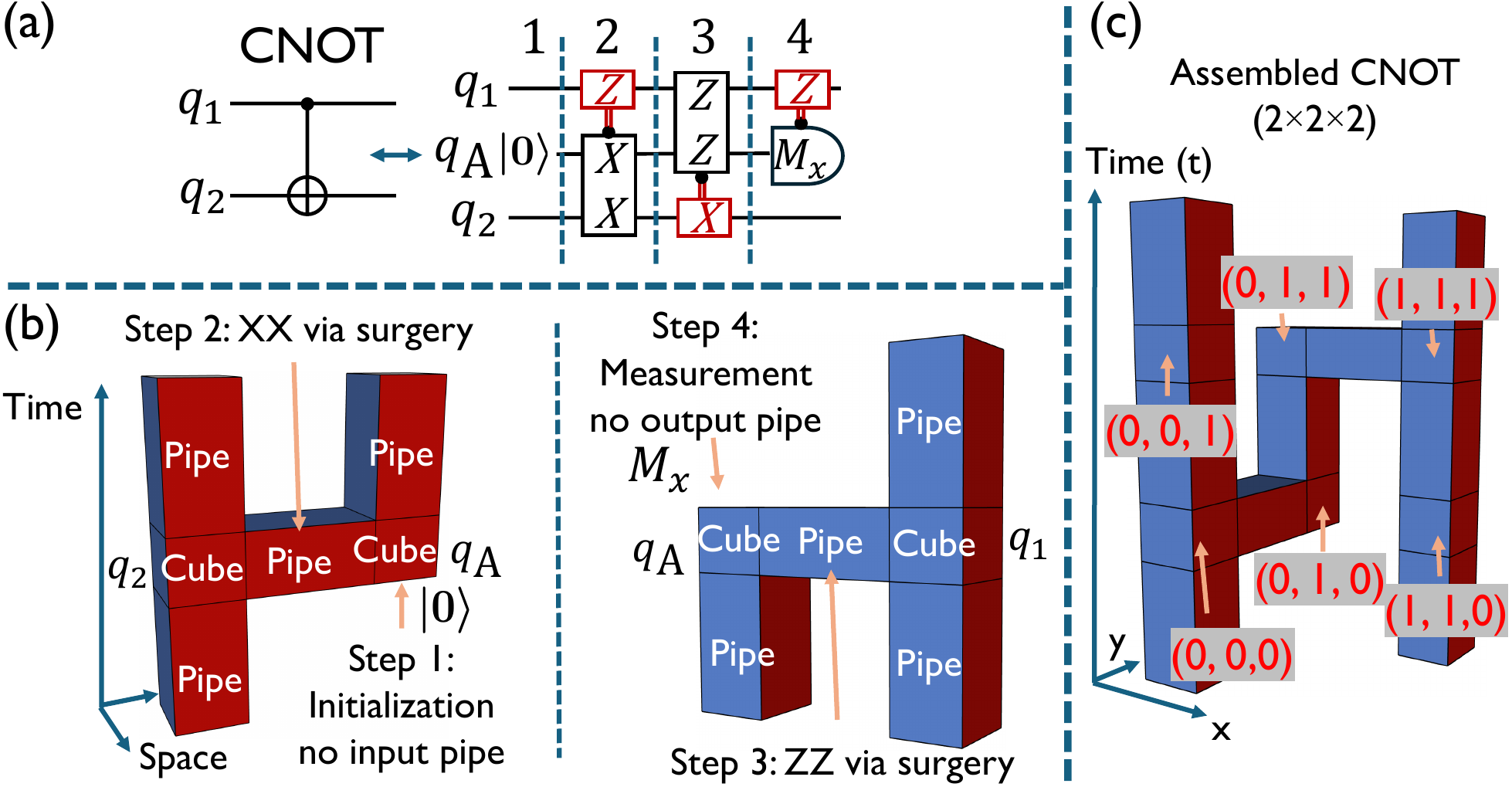}
    \caption{Pipe diagram for logical CNOT gate}
    \label{fig:pipe_cnot}
\end{figure}

\begin{wrapfigure}{R}{0.5\textwidth}
    \centering
    \includegraphics[width=0.5\textwidth]{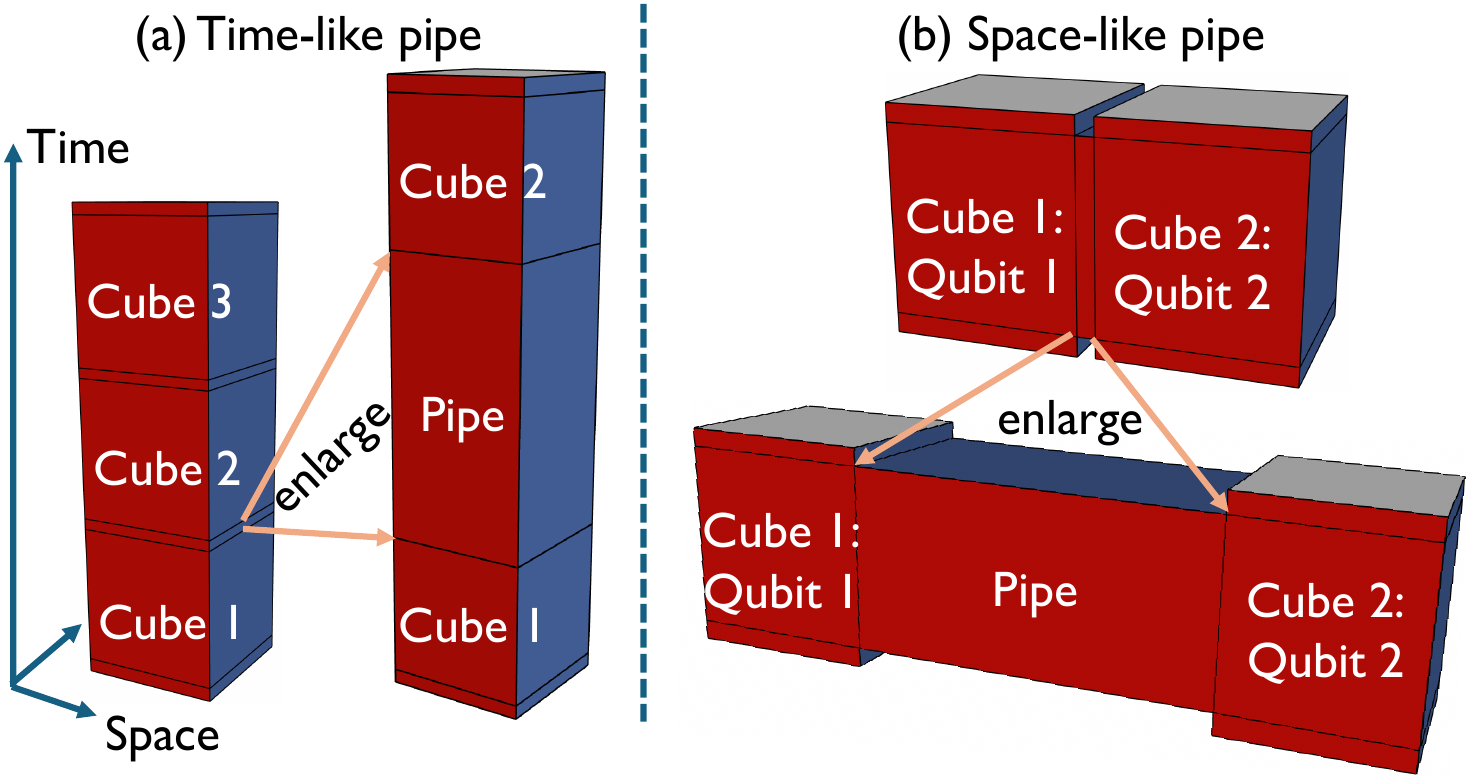}
    % \vspace{-5pt}
    \caption{Pipe diagram}
    % \vspace{-5pt}
    \label{fig:PipeDiagram}
\end{wrapfigure}

We use pipe diagrams~\cite{tan2024sat}, a standard space-time representation for logical surface code patch trajectories and interactions. Figure~\ref{fig:PipeDiagram} derives pipes from the colored block abstraction by enlarging short interfaces between adjacent blocks so connectivity is easier to inspect. Time-like pipes connect consecutive QEC blocks of the same logical patch (Figure~\ref{fig:PipeDiagram}(a)); space-like pipes encode lattice surgery connections between neighboring logical patches (Figure~\ref{fig:PipeDiagram}(b)). This enlargement is purely representational: pipe length is not physical distance or extra qubit overhead. A pipe diagram therefore specifies a lattice-surgery program: where and when each logical patch exists and which patches are merged. It is the target language of \myCompilerName.

Figure~\ref{fig:pipe_cnot}(a) illustrates the four-step measurement-based logical CNOT construction used in this paper~\cite{fowler2018low}. An ancilla logical qubit \(q_A\) is initialized in \(|0\rangle\); the CNOT performs an \(XX\) joint measurement between \(q_2\) and \(q_A\), a \(ZZ\) joint measurement between \(q_1\) and \(q_A\), and an \(X\)-basis measurement of \(q_A\). The red outcome-dependent Pauli corrections are tracked in software via Pauli-frame updates rather than applied physically. Tracking the logical Pauli operators through the three measurements shows that this is a CNOT with control $q_1$ and target $q_2$: they map $X_{q_1}$ to $X_{q_1}X_{q_2}$ and $Z_{q_2}$ to $Z_{q_1}Z_{q_2}$, and leave $Z_{q_1}$ and $X_{q_2}$ unchanged.

Figure~\ref{fig:pipe_cnot}(b) shows the corresponding pipe diagram. Time-like pipes connect consecutive QEC blocks of the same logical qubit, and space-like pipes represent lattice-surgery joint measurements. In Step~2, the \(XX\) measurement becomes a space-like pipe from \(q_2\)'s cube to the newly initialized \(q_A\) cube. Since \(q_A\) is initialized immediately before Step~2, it has no incoming time-like pipe, while \(q_2\)'s time-like pipe extends through the step because \(q_2\) exists before and after the surgery.

Step~3 represents the \(ZZ\) measurement between \(q_1\) and \(q_A\) with another space-like pipe. Step~4 measures \(q_A\) in the \(X\) basis, terminating the ancilla and leaving no outgoing time-like pipe.

Figure~\ref{fig:pipe_cnot}(c) assembles this CNOT in a discrete space-time grid: horizontal axes $x$ and $y$ are spatial patch locations, and the vertical axis $t$ is time; the triples are the $(x,y,t)$ coordinates of the cells. Each cube occupies one grid cell, and each pipe is a unit-length connection between adjacent cells.  The assembled CNOT occupies a \(2 \times 2 \times 2\) \textbf{space-time volume}:
\begin{equation*}
     V = (x_{\max}-x_{\min}+1)\times(y_{\max}-y_{\min}+1)\times(t_{\max}-t_{\min}+1).
\end{equation*}
 This volume estimates the scheduling footprint and physical resources of the lattice-surgery program.

\subsection{ZX-calculus} \label{Sec:ZXCalculus}

ZX-calculus is a graphical language for quantum computation~\cite{coecke2011interacting}. A ZX diagram is a network of ZX nodes and edges denoting a linear map on quantum states. With fixed input and output boundary wires, a diagram $G$ with $n$ inputs and $m$ outputs defines

\[
\llbracket G \rrbracket : (\mathbb{C}^2)^{\otimes n} \rightarrow (\mathbb{C}^2)^{\otimes m}.
\]
Diagrams compose sequentially by connecting outputs to inputs, which multiplies their linear maps, and in parallel by placing them side by side, which takes the tensor product of their linear maps; the map of a diagram is obtained by contracting the tensors of its nodes along its wires. Figure~\ref{fig:ZX}(a) gives the linear map of each node, and Appendix~\ref{Sec:FormalZX} gives the formal semantics. The ZX-calculus also comprises sound rewrite rules on these diagrams, such as spider fusion, identity removal, the bialgebra rule, and the Hopf rule, some of which hold only up to a non-zero scalar; \myCompilerNameSpace uses only spider fusion, its inverse, and identity insertion or removal, which hold exactly (Section~\ref{Sec:TopoLSIR}).

% \begin{figure}[t]
%     \centering
%     \includegraphics[width=0.7\linewidth]{fig/ZX.pdf}
%     % \vspace{-5pt}
%     \caption{ZX-calculus examples}
%     % \vspace{-5pt}
%     \label{fig:ZX}
% \end{figure}

\textbf{ZX nodes.} ZX diagrams are built from typed nodes connected by edges. The main node types used here include phase-labeled $Z$-spiders (white), phase-labeled $X$-spiders (black), and Hadamard nodes (small squares), as shown in Figure~\ref{fig:ZX}(a). A $Z$-spider with phase $\alpha$, $m$ input wires, and $n$ output wires denotes the linear map $|0\rangle^{\otimes n}\langle 0|^{\otimes m}+e^{i\alpha}|1\rangle^{\otimes n}\langle 1|^{\otimes m}$; an $X$-spider is defined in the same way with $|+\rangle$ and $|-\rangle$ in place of $|0\rangle$ and $|1\rangle$; and a Hadamard node denotes $H$. A spider's right and left wires match the ket and bra qubits of its tensor network. For example (Figure~\ref{fig:ZX}(a)), a phase-$0$ $Z$-spider with one output and no input denotes $|0\rangle+|1\rangle=\sqrt{2}|+\rangle$, i.e., initialization in $|+\rangle$ up to a scalar; a phase-$0$ $X$-spider with one input and no output denotes $\langle +|+\langle -|=\sqrt{2}\langle 0|$, i.e., the outcome-$0$ branch of a $Z$-basis measurement up to a scalar; and a phase-$0$ spider with one input and one output denotes the identity.

\begin{wrapfigure}{R}{0.67\textwidth}
    \centering
    \includegraphics[width=0.67\textwidth]{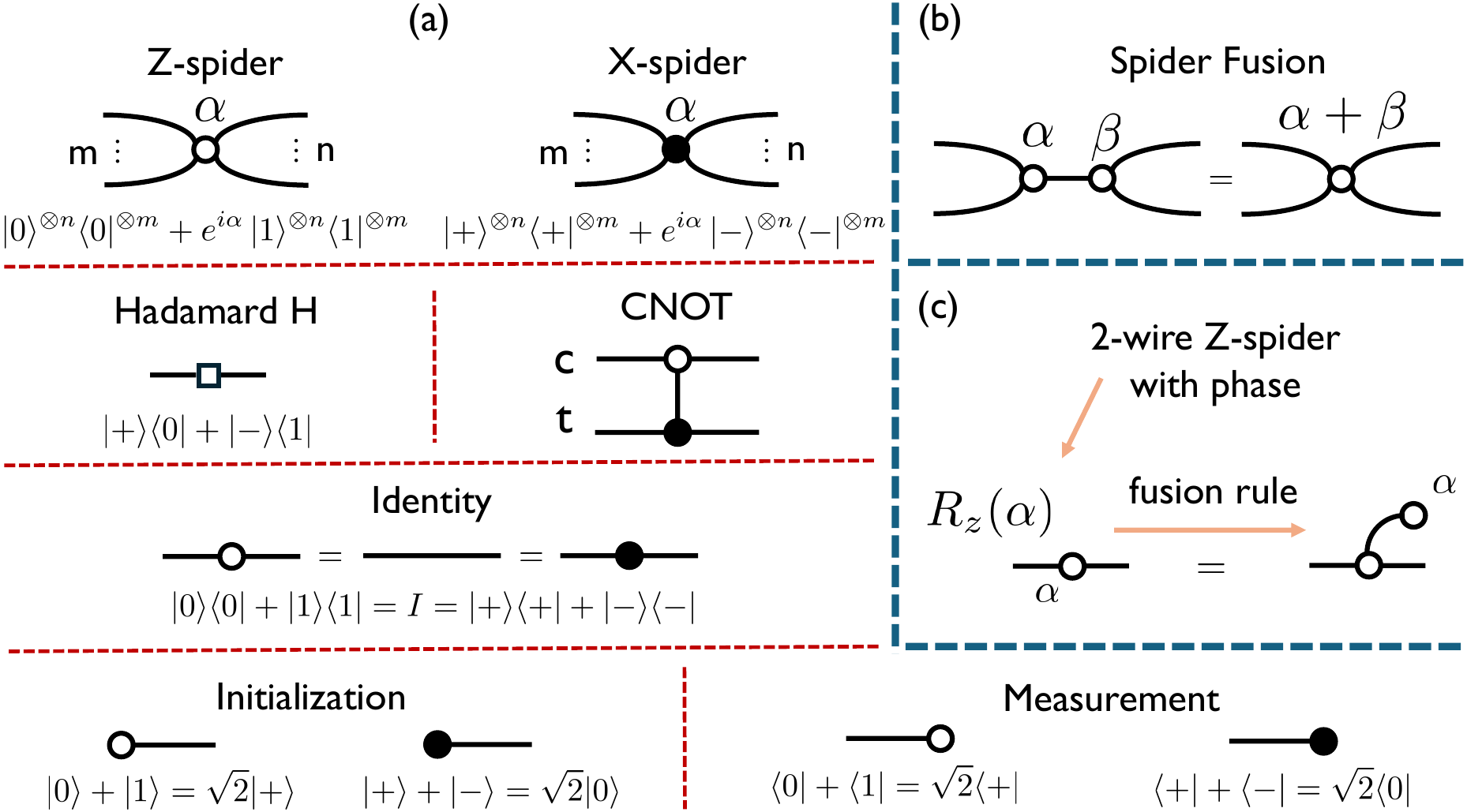}
    % \vspace{-5pt}
    \caption{ZX-calculus examples}
    % \vspace{-5pt}
    \label{fig:ZX}
\end{wrapfigure}

\textbf{Spider Fusion.} Same-type spiders can merge into one spider whose phase is the sum of the original phases (Figure~\ref{fig:ZX}(b)); this semantic-preserving rule enables algebraic simplification.

\textbf{Quantum gates in ZX-calculus.}
Common gates are built from the ZX nodes in Figure~\ref{fig:ZX}(a) and spider fusion in Figure~\ref{fig:ZX}(b). In Figure~\ref{fig:ZX}(c), a two-wire \(Z\)-spider with phase \(\alpha\) denotes \( |0\rangle\langle 0| + e^{i\alpha}|1\rangle\langle 1| \), implementing \(R_z(\alpha)\); \(S\) and \(T\) are the cases \(\alpha=\pi/2\) and \(\alpha=\pi/4\). Equivalently, \(R_z(\alpha)\) is an identity \(Z\)-spider fused with an auxiliary same-phase node. CNOT (Figure~\ref{fig:ZX}(a)) is a \(Z\)-spider on the control wire connected by an internal edge to an \(X\)-spider on the target wire. The internal edge is drawn vertically; this is unambiguous because only the connectivity of a ZX diagram matters (Section~\ref{Motivation of TopoLS}), so the edge may be bent forward or backward in time without changing the map.

\subsection{Monte Carlo Tree Search} \label{MonteCarloTreeSearch}
We use Monte Carlo Tree Search (MCTS)~\cite{browne2012survey, kocsis2006bandit} later to optimize pipe-diagram space-time volume.

MCTS is a heuristic search algorithm for large decision spaces. It incrementally builds a tree by simulating future outcomes while balancing exploration of less-visited actions and exploitation of high-value actions. Each iteration has four stages: \textbf{Selection}, which uses a policy such as Upper Confidence Bound for Trees (UCT) to choose a promising node; \textbf{Expansion}, which adds child nodes; \textbf{Simulation}, which estimates outcomes through random or policy-guided playouts; and \textbf{Backpropagation}, which propagates simulation results to update value estimates.

% \begin{figure}[t]
%     \centering
%     \includegraphics[width=0.7\linewidth]{fig/MCTS.pdf}
%     % \vspace{-5pt}
%     \caption{Overview of Monte Carlo tree search}
%     % \vspace{-5pt}
%     \label{fig:MCTS}
% \end{figure}

MCTS can be stopped after any number of iterations and then returns the best solution found so far, so compilation time can be traded for solution quality.

\section{The \myCompilerNameSpace Intermediate Representation} \label{Sec:TopoLSIR}

The central abstraction in \myCompilerNameSpace is the use of typed open ZX diagrams, i.e., ZX diagrams with dangling input and output wires, as the compiler IR for lattice surgery. The source program is a quantum circuit, the IR is a \emph{boundary-indexed ZX graph}: an open ZX graph with indexed, qubit-labeled input and output boundary wires, and the target is a valid pipe diagram realizing the same denotation. This organization separates semantic optimization from physical realization: ZX transformations and analyses optimize the IR while preserving its denotation, and the backend searches for a valid pipe realization with reduced space-time volume.

\subsection{Motivation: Why ZX as an IR?} \label{Motivation of TopoLS}

\begin{wrapfigure}{R}{0.34\textwidth}
    \centering
    \includegraphics[width=0.34\textwidth]{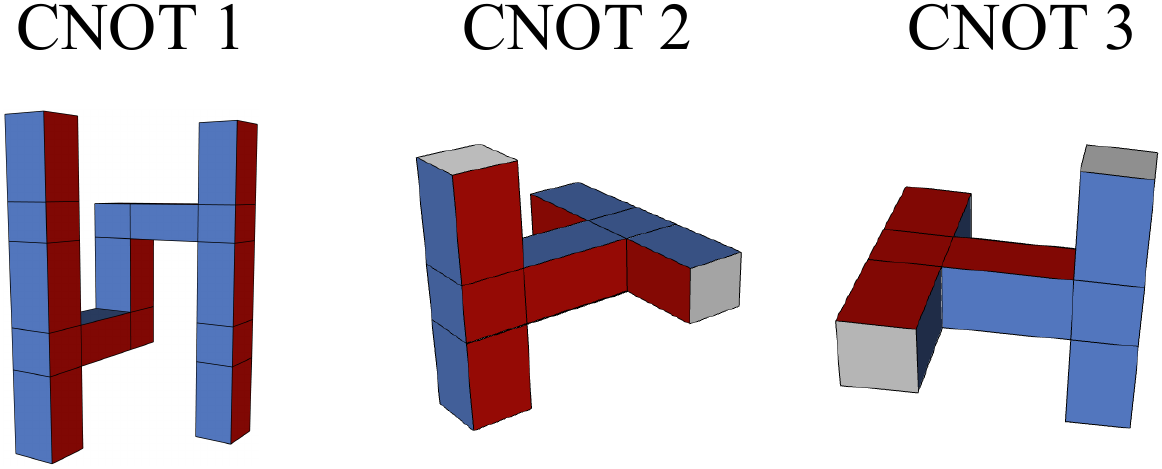}
    \captionsetup{justification=raggedright,singlelinecheck=false}
    \caption{Different implementations of CNOT}
    \label{fig:CNOT_topo}
\end{wrapfigure}

\textbf{Characteristics of Lattice Surgery and Limitation of Circuit-Centric Compilation.}
Lattice surgery programs can be represented by pipe diagrams, where cubes denote operation primitives and pipes denote topological connections between them. Computations are therefore determined by the connectivity of cubes and pipes rather than by a fixed sequence of gates. As a result, the same logical computation can often be realized through multiple spatial and temporal pipe layouts.
Figure~\ref{fig:CNOT_topo} illustrates that three different pipe layouts implement the same CNOT operation. However, existing circuit-centric compilers typically derive lattice surgery implementations directly from gate-based circuits with predetermined layouts. As a result, each CNOT is mapped to a fixed canonical layout (e.g., CNOT 1 in Figure~\ref{fig:CNOT_topo}), while alternative layouts with smaller space--time volume are not explored.

\textbf{ZX-calculus as an Intermediate Representation.}
ZX-calculus provides a graphical language for representing quantum computations~\cite{duncan2020graph,kissinger2019reducing,peham2022equivalence}. 
Recent work has shown that ZX diagrams can represent lattice surgery operations~\cite{de2020zx}, assist in designing fault-tolerant layouts~\cite{gidney2024inplace}, and verify lattice surgery constructions~\cite{tan2024sat}. 
In ZX diagrams, computations are represented as graphs of spiders and edges that capture only the connectivity between operations.

\begin{wrapfigure}{R}{0.3\textwidth}
    \centering
    \includegraphics[width=0.3\textwidth]{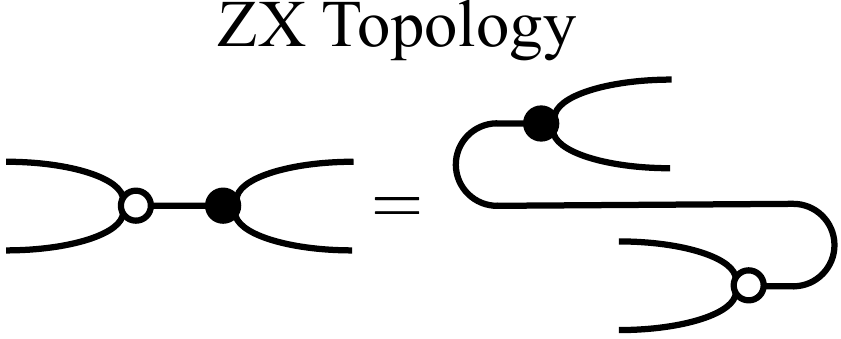}
    \caption{Topological property of ZX diagram}
    \label{fig:ZX_topo}
\end{wrapfigure}

Importantly, ZX diagrams are invariant under geometric deformation: wires can be stretched, bent, or rearranged without changing the represented computation, an example is shown in Figure~\ref{fig:ZX_topo}: drawing the $X$ spider before the $Z$ spider and bending the wire between them into a cup and a cap keeps the connectivity and therefore the linear map.
This topological property aligns naturally with pipe diagrams: the semantics depend on which pipe-diagram components are connected, while the geometric route used to realize the connection can vary.
In particular, the blue and red cubes shown in Figure~\ref{fig:ZX_Pipe}(a) correspond to $Z$ and $X$ spiders, respectively, and pipes between cubes correspond to ZX edges between the associated spiders; we give the complete pipe-to-ZX correspondence formally in Section~\ref{Sec:SourceTargetTranslations}.

Figure~\ref{fig:ZX_Pipe}(b) further illustrates this mapping using a CNOT example: translating the cubes and pipes in the pipe diagram into spiders and ZX edges produces a ZX graph that is equivalent to the ZX representation of the original CNOT circuit.

\begin{wrapfigure}{R}{0.55\textwidth}
    \centering
    \includegraphics[width=0.55\textwidth]{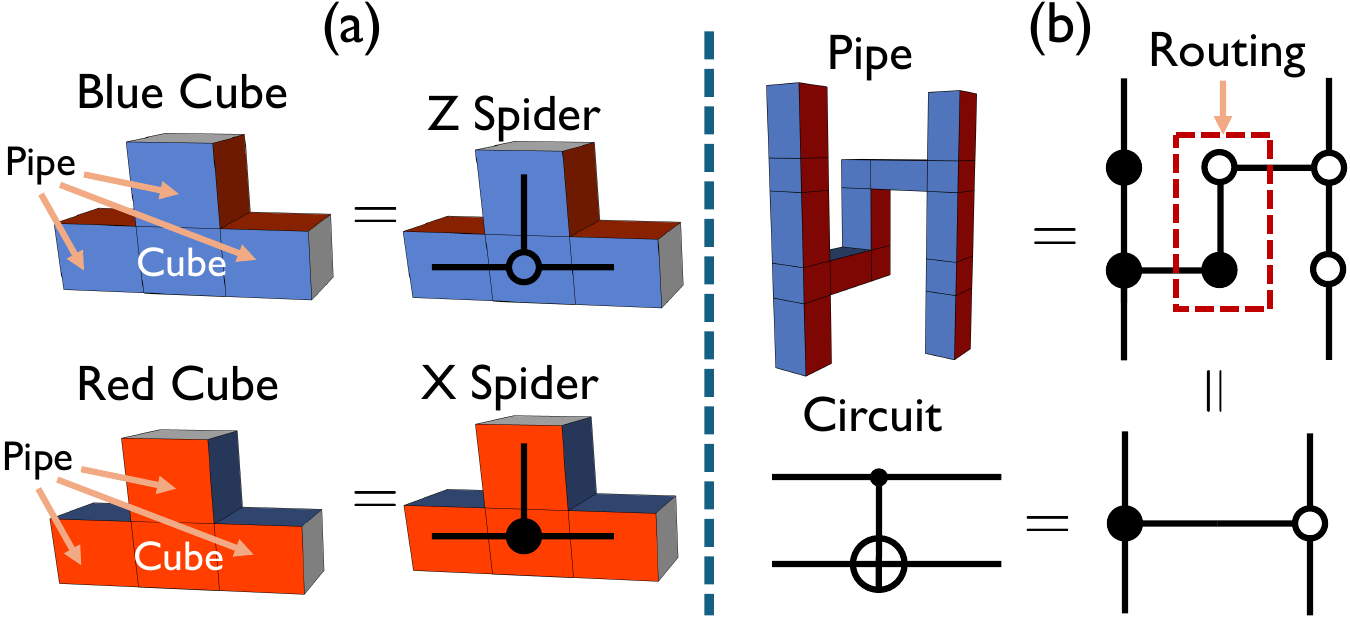}
    \caption{Pipe diagram to ZX-calculus}
    \label{fig:ZX_Pipe}
\end{wrapfigure}

With this correspondence, we can revisit the multiple implementations of a CNOT operation shown earlier in Figure~\ref{fig:CNOT_topo} from the perspective of ZX diagrams.
In ZX representation, a CNOT is represented by a $Z$ spider connected to an $X$ spider.

Since ZX diagrams specify only the connectivity between spiders and abstract away geometric details, any connected route in 3D space between the corresponding cubes represents the same computation.
In the pipe diagram, this corresponds to connecting the blue and red cubes by a routed connection made of one or more pipes, allowing multiple pipe layouts to implement the same CNOT operation. From the compiler's perspective, selecting a feasible routed connection between the corresponding cubes becomes a \textbf{routing problem}. A blue or red cube incident to exactly two pipes corresponds to a degree-two, zero-phase spider, which is semantically an identity wire in ZX-calculus. Such cubes can therefore be inserted as \textbf{routing intermediates}: they may bend or extend a routed connection without changing the represented computation. In Figure~\ref{fig:ZX_Pipe}(b), the middle $Z$ and $X$ spiders, corresponding to the intermediate blue and red cubes in the pipe diagram, serve precisely this routing role.

\subsection{IR Syntax and Denotation} \label{Sec:IRSyntax}

\begin{definition}[\myCompilerNameSpace ZX IR] \label{def:TopoLSIR}
A \myCompilerNameSpace IR program is a boundary-indexed ZX graph represented as a tuple
\[
G=(V_G,E_G,\partial_{\mathrm{in}}G,\partial_{\mathrm{out}}G,\kappa_G,\varphi_G).
\]
Here $V_G$ is the set of ZX-generator vertices, $E_G$ is the set of ZX-wire edges, $\partial_{\mathrm{in}}G$ and $\partial_{\mathrm{out}}G$ are the indexed input and output boundary wires whose indices record the corresponding logical qubits, $\kappa_G(v)$ gives the generator kind of a vertex $v$, and $\varphi_G(v)$ gives its phase when applicable.
\end{definition}

Concretely, $G$ is the graph representation that PyZX~\cite{kissinger2019pyzx} already provides, whose ordered input and output lists give the boundary indices; \myCompilerNameSpace does not define a new diagrammatic language. Its passes read and rewrite this representation, and slicing and partitioning also use the qubit and row annotations that PyZX keeps for each vertex of a translated circuit.

The generators can exactly describe all the lattice-surgery realizations shown in Figure~\ref{fig:Gate}: $Z$ spiders, $X$ spiders, Hadamard nodes, and phase/resource nodes.

\begin{definition}[IR Denotation and Equivalence] \label{def:IRDenotation}
The denotation of an IR graph $G$, written $\llbracket G \rrbracket$, is the standard linear-map semantics of ZX-calculus, with $\partial_{\mathrm{in}}G$ and $\partial_{\mathrm{out}}G$ determining the input and output wire order. Two IR graphs $G$ and $G'$ are semantically equivalent when $\llbracket G \rrbracket=\llbracket G' \rrbracket$.
\end{definition}

This equivalence is intentionally semantic rather than syntactic: sound ZX rewrites and the insertion of degree-two zero-phase identity spiders may change the graph representation without changing the represented linear map.
Exact equality is appropriate because every rewrite that \myCompilerNameSpace applies (spider fusion, its inverse, and the insertion of degree-two zero-phase spiders) preserves the denotation exactly; other ZX rules and the translation of a circuit (e.g., $\sqrt{2}$ per CNOT) may introduce a non-zero global scalar, which does not change the computation.

\subsection{Source and Target Translations}\label{Sec:SourceTargetTranslations}

\begin{definition}[Circuit-to-IR Translation] \label{def:CircuitToIR}
The frontend translates a source circuit $C$ into a ZX IR graph
\[
\mathcal{T}_C : C \rightarrow G_C .
\]
The translation maps circuit operations to ZX generators and wires, following the rules in Figure~\ref{fig:ZX}, while preserving the qubit-labeled input and output boundary wires.
\end{definition}

\begin{figure}[H]
    \centering
    \includegraphics[width=0.8\linewidth]{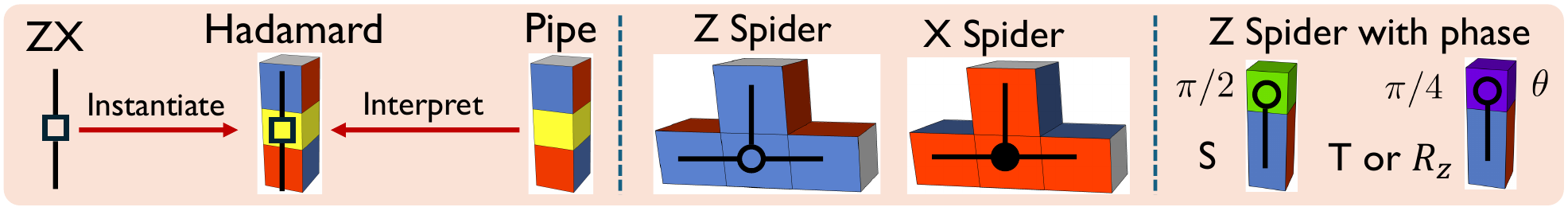}
    \caption{Instantiation and interpretation between ZX diagrams and pipe primitives}
    \label{fig:Gate}
\end{figure}

\begin{definition}[Pipe-to-IR Interpretation] \label{def:PipeToIR}
The backend interprets a pipe diagram $P$ as a ZX graph
\[
\mathcal{T}_P : P \rightarrow G_P .
\]
The interpretation follows Figure~\ref{fig:Gate}: cubes are interpreted as ZX nodes and pipes are interpreted as ZX edges. Blue cubes correspond to zero-phase $Z$ spiders, red cubes correspond to zero-phase $X$ spiders, and yellow cubes realize Hadamard nodes that switch the color of pipes. The green cube, corresponding to a $Z$ spider with phase $\pi/2$, is used to implement an $S$ gate. The purple cube is a magic-state source, corresponding to a $Z$ spider with phase $\pi/4$ or $\theta$ used to implement a $T$ gate or an $R_z(\theta)$ gate. We explain the physical details in Appendix~\ref{Sec:Gate}.
Conversely, when the compiler starts from a ZX node and its incident wires and constructs the corresponding cubes and pipes, we call this process \emph{instantiation}; this direction is used by the backend in Section~\ref{Sec:ZXToPipeInstantiation}. The interpretation and instantiation directions are shown together in the left part of Figure~\ref{fig:Gate}.
\end{definition}

\subsection{Pipe Realization Contract} \label{Sec:PipeRealizationContract}

% \begin{wrapfigure}{R}{0.4\textwidth}
%     \centering
%     \includegraphics[width=0.4\textwidth]{fig/Pipe_rep.pdf}
%     \caption{Pipe diagram variables}
%     \label{fig:pip_rep}
% \end{wrapfigure}

\begin{figure}[t]
    \centering
    \includegraphics[width=0.9\textwidth]{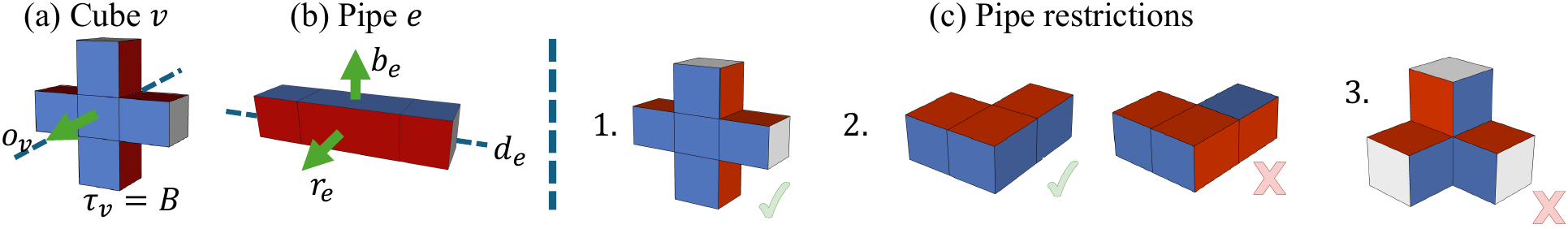}
    \caption{Pipe diagram variables}
    \label{fig:pip_rep}
\end{figure}

\begin{definition}[Pipe Diagram] \label{def:PipeDiagramRepresentation}\label{def:ValidPipeDiagram}
A pipe diagram is a finite space--time object
\[
P=(C,E,\tau,p,o,b,r),
\]
where $C$ is a finite set of cubes, $E$ is a finite set of pipe segments between cubes, and
\[
\mathsf{Kind}=\{\mathsf{B},\mathsf{R},\mathsf{Y},\mathsf{G},\mathsf{P}\},
\qquad
\mathsf{Axis}=\{x,y,z\}.
\]
Here $\mathsf{B},\mathsf{R},\mathsf{Y},\mathsf{G}$, and $\mathsf{P}$ denote the blue, red, yellow, green, and purple cube colors, respectively.
The cube-level functions are
\[
\tau:C\to\mathsf{Kind},\qquad
p:C\to\mathbb{Z}^3,\qquad
o:C\rightharpoonup\mathsf{Axis},
\]
where $\tau$ assigns cube types, $p$ assigns lattice positions, and $o$ assigns orientation axes only for blue and red cubes. The pipe-level functions are
\[
b:E\to\mathsf{Axis},
\qquad
r:E\to\mathsf{Axis},
\]
which assign the orientations of the blue and red boundaries of each pipe segment. We use subscript notation throughout: $\tau_v$, $p_v$, and $o_v$ for cube attributes (Figure~\ref{fig:pip_rep}(a)), and $b_e$ and $r_e$ for pipe attributes (Figure~\ref{fig:pip_rep}(b)). For a pipe $e=(u,v)$, its direction is the derived vector
\[
d_e=p_v-p_u\in\{\pm\hat{x},\pm\hat{y},\pm\hat{z}\}.
\]
Input/output ports are ordinary blue or red cubes on the open boundary of the diagram.

We write $P\in\mathcal{P}_{\mathrm{valid}}$ when the following conditions hold. See Figure~\ref{fig:pip_rep}(c) for examples and counterexamples.

\begin{enumerate}[label=(V\arabic*), leftmargin=*]
\item \textbf{Occupancy and ports.} The position map $p$ is injective, so no two cubes occupy the same lattice cell. The open input/output ports of $P$ are ordinary blue or red cubes and are ordered consistently with the indexed boundary wires of the IR graph being realized.
\item \textbf{Pipe geometry.} Every pipe $e=(u,v)\in E$ connects adjacent lattice cells: $\|p_u-p_v\|_1=1$. Pipe interiors are disjoint from all other pipe interiors and from all cube cells except their endpoints.
\item \textbf{Blue/red cubes.} If $\tau_v\in\{\mathsf{B},\mathsf{R}\}$, then every incident pipe $e$ satisfies $d_e\cdot\hat{o}_v=0$ and $\deg(v)\leq 4$. Moreover, the pipe color must be consistent with the incident blue or red cube: for a blue cube, every incident pipe $e$ has $b_e=o_v$; for a red cube, every incident pipe $e$ has $r_e=o_v$.
\item \textbf{Yellow cubes.} If $\tau_v=\mathsf{Y}$, then $\deg(v)=2$. For its two incident pipes $e_1,e_2$, the color orientations are swapped across the cube: $b_{e_1}=r_{e_2}$ or $r_{e_1}=b_{e_2}$.
\item \textbf{Green/purple cubes.} If $\tau_v\in\{\mathsf{G},\mathsf{P}\}$, then $\deg(v)=1$.
\end{enumerate}
\end{definition}

% \begin{figure}[t]
%     \centering
%     \includegraphics[width=0.6\linewidth]{fig/Pipe_R.pdf}
%     % \vspace{-5pt}
%     \caption{Pipe restrictions}
%     \vspace{-5pt}
%     \label{fig:Pipe_R}
% \end{figure}

\begin{definition}[Pipe Realization of an IR Graph] \label{def:PipeRealization}
A pipe diagram $P$ realizes an IR graph $G$ when
\[
\llbracket \mathcal{T}_P(P) \rrbracket = \llbracket G \rrbracket
\qquad\text{and}\qquad
P\in\mathcal{P}_{\mathrm{valid}} .
\]
\end{definition}

The realization relation is semantic rather than syntactic. The interpreted graph $\mathcal{T}_P(P)$ may differ from $G$ by sound ZX rewrites and by inserted degree-two zero-phase spiders used as routing intermediates. These routing intermediates are identity wires in ZX-calculus, so they may change the geometric shape of a routed connection without changing the denotation of the program.

% \begin{wrapfigure}{R}{0.6\textwidth}
%     \centering
%     \includegraphics[width=0.6\textwidth]{fig/Pipe_R.pdf}
%     \caption{Pipe restrictions}
%     \label{fig:pip_rep}
% \end{wrapfigure}

\subsection{IR Contract for Compilation Passes}

\begin{principle}[IR-Preserving Compilation]\label{prin:IRPreservingCompilation}
\myCompilerNameSpace compilation is organized around the IR contract in Definition~\ref{def:TopoLSIR}. Given the ZX-translated graph $\mathcal{T}_C(C)$ of the original quantum circuit $C$, a pass is correct if it either transforms an IR graph $G$ into an equivalent IR graph $G'$ with $\llbracket G'\rrbracket=\llbracket G\rrbracket$, or produces a valid pipe diagram $P$ whose interpretation realizes the current IR graph, i.e., $\llbracket \mathcal{T}_P(P)\rrbracket=\llbracket G\rrbracket$. Section~\ref{Sec:Optimizationat} instantiates this contract with ZX-level optimizations and backend layout search.
\end{principle}

This organization is the main reason for using ZX-calculus as an IR rather than only as a proof language: the same representation supports semantic equivalence, graph-level optimization, and correctness-preserving lowering to lattice surgery.

\section{Optimization Opportunities and Compilation Problem} \label{Sec:ProblemFormulation}

Section~\ref{Sec:TopoLSIR} introduced the ZX-based IR and motivated why it is useful for lattice surgery compilation. This section states the optimization opportunities and concrete compilation problem solved by \myCompilerName: synthesize a valid pipe-level realization of the IR with small space--time volume.

\subsection{Optimization Opportunities Exposed by the IR} \label{Sec:Optimization Opportunities}

\begin{wrapfigure}{R}{0.6\textwidth}
    \centering
    \includegraphics[width=0.6\textwidth]{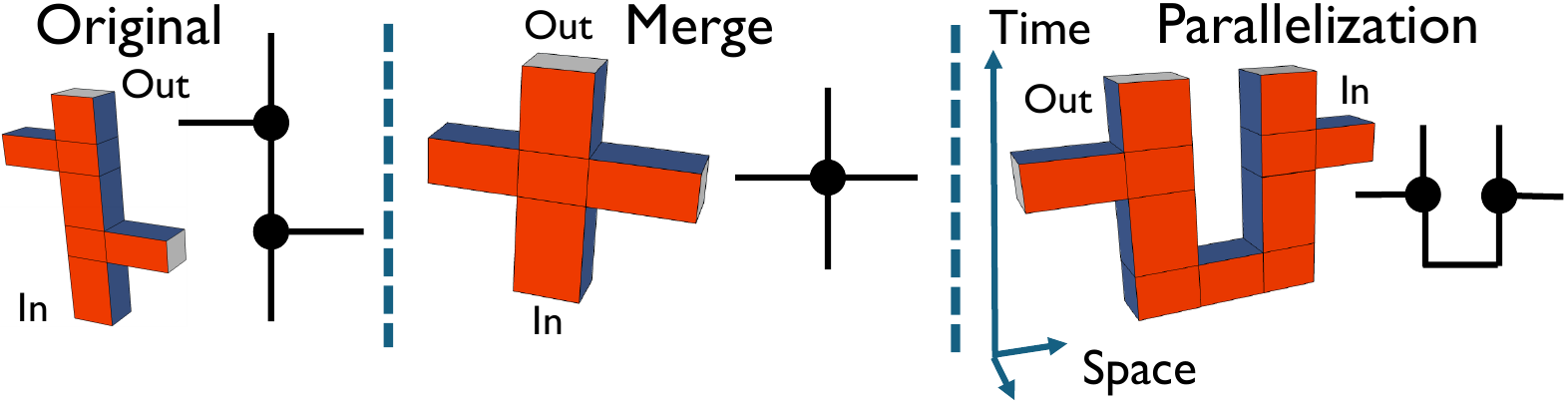}
    \caption{Optimization opportunity from ZX representation. In the right panel, ``In'' remains an input: inputs and outputs are fixed by boundary indices (Definition~\ref{def:TopoLSIR}), not by wire direction.}
    \label{fig:ZX_opp}
\end{wrapfigure}

Once the program is represented as a ZX graph, compiler optimizations can operate on topological connectivity before committing to a concrete pipe layout. This IR view exposes two opportunities that are difficult to express in a gate sequence.

% \begin{wrapfigure}{R}{0.5\textwidth}
%     \centering
%     \includegraphics[width=0.5\textwidth]{fig/gate.pdf}
%     % \vspace{-5pt}
%     \caption{Lattice surgery pipe to ZX graph}
%     % \vspace{-5pt}
%     \label{fig:Gate}
% \end{wrapfigure}

% \textbf{First}, the ZX diagram provides a representation of the program at the topological level, enabling optimizations that are not easily accessible in the gate-based model. For example, spider fusion can be applied to combine merge--split operations, as illustrated in middle of Figure~\ref{fig:ZX_opp}, or the execution order of merge--split operations can be rearranged to enhance parallelism, see right of Figure~\ref{fig:ZX_opp}. Such optimizations arise from the fact that lattice surgery, like to ZX diagrams, is agnostic to routing details; consequently, structural transformations applied at the ZX-diagram level directly translate into high-level modifications of the lattice surgery protocol.

% \begin{wrapfigure}{R}{0.7\textwidth}
%     \centering
%     \includegraphics[width=0.7\textwidth]{fig/2D3D.pdf}
%     \caption{2D routing vs 3D routing}
%     \label{fig:2D3D}
% \end{wrapfigure}

\textbf{First}, the ZX diagram provides a representation of the program at the topological level, enabling optimizations that are difficult to expose in the gate-based model. For example, spider fusion can combine multiple pipe-diagram cubes when their corresponding ZX spiders have the same type, as illustrated in the middle of Figure~\ref{fig:ZX_opp}.
ZX diagrams can also expose parallelism that is hidden by a gate-based sequence. Since both ZX diagrams and lattice surgery are defined by topological connectivity rather than fixed routed connections, the compiler can first transform the ZX graph and then choose how to realize the corresponding cubes and pipes in space--time. ZX diagrams themselves do not impose an execution order; in this paper, \emph{execution order} refers to this compiler-chosen space--time realization. Thus, transformations such as selecting a more parallel realization of the same ZX connectivity can yield more efficient lattice surgery protocols without changing the computation, as shown on the right of Figure~\ref{fig:ZX_opp}.

% \begin{figure}[t]
%     \centering
%     \includegraphics[width=0.7\linewidth]{fig/ZX_opp.pdf}
%     % \vspace{-5pt}
%     \caption{Optimization opportunity from ZX representation}
%     % \vspace{-5pt}
%     \label{fig:ZX_opp}
% \end{figure}

% \textbf{Second}, the topological nature of lattice surgery allows us to explore a much larger solution space for the routing problem. Prior work has primarily addressed routing by considering only a 2D slice of the quantum circuit, overlooking the inherently 3D structure and topological properties of lattice surgery. This distinction is illustrated in Figure~\ref{fig:2D3D}. In the example circuit with two CNOT gates, restricting routing to the 2D circuit slice (where yellow nodes denote data qubits and green nodes denote ancilla qubits) prevents both CNOTs from being executed simultaneously. In contrast, when routing is extended into the full 3D space, spare regions from earlier execution steps can be reused to establish the necessary topological connections, enabling more efficient scheduling.

\begin{figure}[t]
    \centering
    \includegraphics[width=0.65\linewidth]{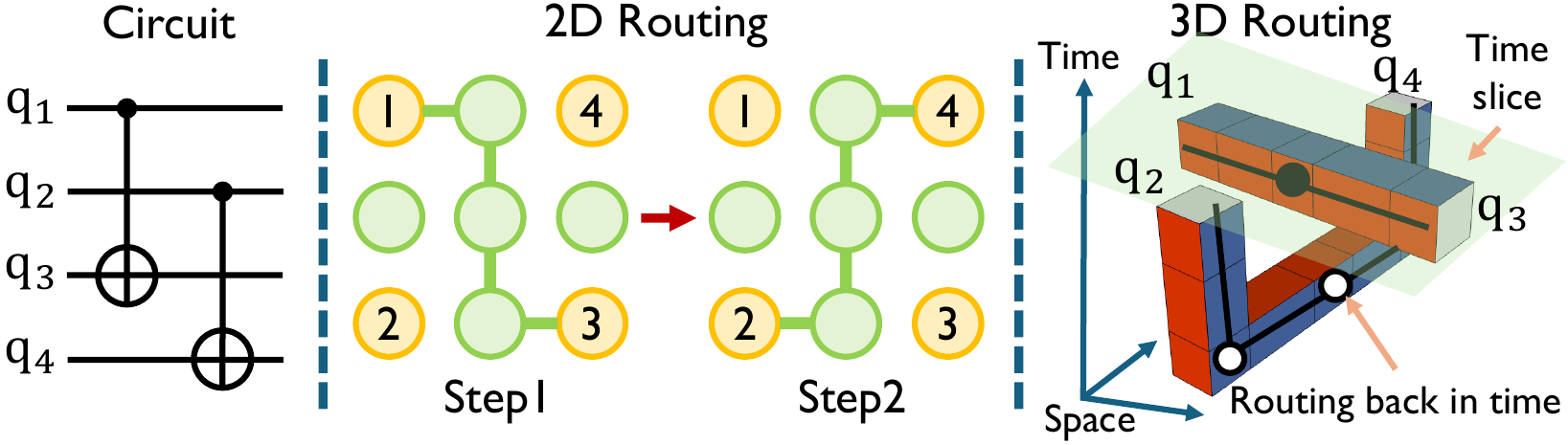}
    % \vspace{-5pt}
    \caption{2D routing vs 3D routing}
    % \vspace{-5pt}
    \label{fig:2D3D}
\end{figure}

\textbf{Second}, the topological nature of lattice surgery allows us to explore a much larger solution space for routing. Figure~\ref{fig:2D3D} illustrates the key distinction. In the \emph{2D routing} view, each step is planned on a single spatial slice: the first CNOT consumes a routing channel in Step~1, and the second CNOT must be routed in a later slice, Step~2, because the two connections conflict in the plane.
By contrast, the pipe-diagram representation treats routing as a three-dimensional space--time problem. The compiler may route a connection through different time slices by connecting pipes with blue/red cubes, including moving forward or backward along the time axis, so that a connection can use space that is blocked in one slice but available in another. In Figure~\ref{fig:2D3D}, the 3D route bends back in time to connect the relevant cubes without requiring the two CNOT routes to occupy the same 2D slice. This additional routing freedom can expose higher parallelism and reduce the resulting space--time volume.

\textbf{Topology-Centric Optimization Example}

\begin{wrapfigure}{R}{0.48\textwidth}
    \centering
    \includegraphics[width=0.48\textwidth]{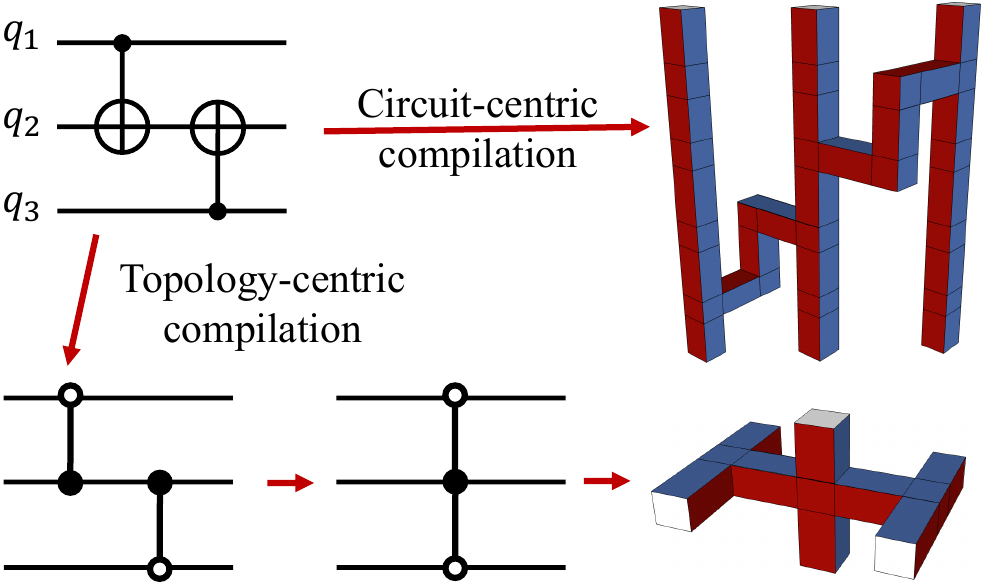}
    \caption{Running example}
    \label{fig:run}
\end{wrapfigure}

We use the example in Figure~\ref{fig:run} to illustrate how topology-centric compilation can produce more efficient lattice surgery layouts than circuit-centric approaches.

Under circuit-centric compilation, each CNOT is implemented using a fixed pipe layout that requires two time steps. Consequently, executing the two CNOTs sequentially takes four time steps and requires two additional ancilla logical qubits.

In contrast, after transforming the circuit into a ZX diagram, the two CNOT operations can be simplified through spider fusion on their target spiders. This transformation reveals a more compact structure, allowing the lattice surgery layout to be executed within a single time step and without introducing additional ancilla logical qubits in the resulting 3D pipe diagram.

% \begin{figure}[t]
%     \centering
%     \includegraphics[width=0.5\linewidth]{fig/Running_exp.pdf}
%     \caption{}
%     \label{fig:run}
% \end{figure}

\subsection{Compilation Objective and Target Constraints} \label{Problem Formulation}

We now state the concrete optimization problem induced by the IR contract in Section~\ref{Sec:TopoLSIR}. Given a source circuit $C$, the compiler first constructs a ZX IR graph $G_C=\mathcal{T}_C(C)$ and then seeks a valid pipe diagram $P$ whose interpretation preserves the IR semantics while minimizing space--time volume.

\begin{definition}[IR Realization Optimization Problem]\label{def:IRRealizationOptimization}
Given a source circuit $C$, let $G_C=\mathcal{T}_C(C)$ be its boundary-indexed ZX IR graph. The TopoLS compilation objective is to find a pipe diagram $P$ solving
\[
\begin{aligned}
\min_{P} \quad & V(P) \\
\text{s.t.} \quad 
& \llbracket \mathcal{T}_P(P) \rrbracket
=
\llbracket G_C \rrbracket, \\
& P \in \mathcal{P}_{\mathrm{valid}} .
\end{aligned}
\]
Here $V(P)$ is the space--time volume of $P$, measured as the volume of the bounding box enclosing all cubes in $P$. The equality constraint is semantic correctness: the pipe diagram, when interpreted back as a ZX graph, has the same denotation as the IR of the source circuit. The validity constraint requires the pipe diagram to satisfy Definition~\ref{def:ValidPipeDiagram}.
\end{definition}

\section{Compiler Passes over the ZX IR} \label{Sec:Optimizationat}

Given the ZX-based IR and compilation objective in Section~\ref{Problem Formulation}, \myCompilerNameSpace is organized as a sequence of IR-preserving passes followed by backend realization.  This section instantiates Principle~\ref{prin:IRPreservingCompilation}: every middle-end pass preserves the ZX denotation, while the backend returns a valid pipe realization of that denotation.

\subsection{ZX-Level IR Transformation} \label{Sec:TopologicalOptimizationat}

\begin{figure}[t]
    \centering
 %   \vspace{3pt}
    \includegraphics[width=1\textwidth]{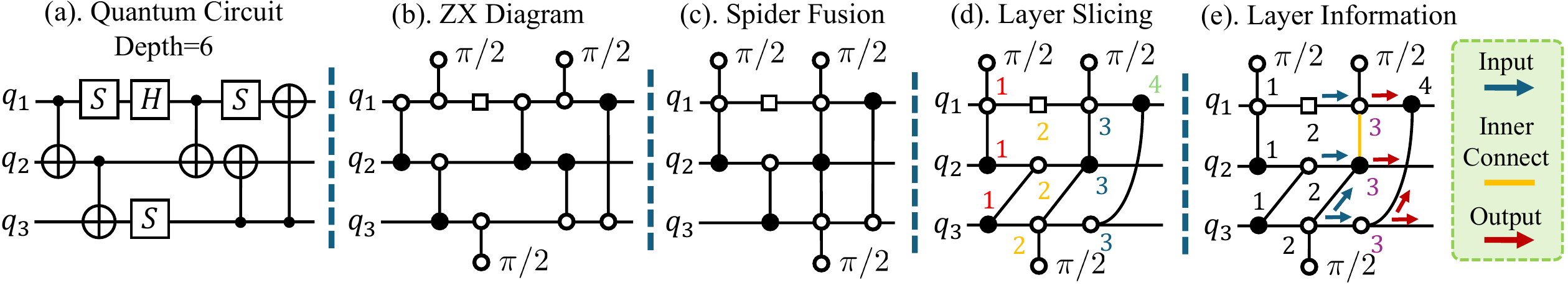}
   % \vspace{-13pt}
    \caption{Topological optimization via ZX representation}
    % \vspace{-12pt}
    \label{fig:ZX_Opt}
\end{figure}

\subsubsection{\textbf{ZX Simplification via Spider Fusion}} \label{Sec:ZXRepresentationandSpiderFusion}

To enable topological optimization, the first step is to transform the quantum circuit into its ZX graph IR using Definition~\ref{def:CircuitToIR}. We illustrate this process with the example in Figure~\ref{fig:ZX_Opt}(a), which has a circuit depth of 6. Although only $H$, $S$, and CNOT gates are shown, other gates can be handled using the same techniques described in this section. The corresponding ZX diagram obtained from this direct transformation is shown in Figure~\ref{fig:ZX_Opt}(b).

Following this transformation, spider fusion acts as the first compiler pass over the IR: it merges adjacent same-color spiders before any pipe path is chosen, thereby changing the future cube-level realization while preserving the ZX denotation, as described in Figure~\ref{fig:ZX_opp}.

% Since spider fusion is a sound rewrite rule in ZX-calculus, the transformation preserves the semantics of the represented computation. 
% Formally, if $G$ and $G_{fuse}$ denote the ZX diagrams before and after spider fusion, then
% \[
% \llbracket G \rrbracket = \llbracket G_{fuse} \rrbracket .
% \]

In the example above, three $Z$ spider fusions and one $X$ spider fusion are performed, with the resulting diagram shown in Figure~\ref{fig:ZX_Opt}(c). When merging spiders, it is important to ensure that the corresponding blue or red cube remains physically realizable: the merged spiders should have at most four wires, corresponding to four pipes incident to the cube, as described by the blue/red cube constraint in Definition~\ref{def:ValidPipeDiagram}.

The correctness of this transformation follows from the soundness of the spider fusion rule in ZX-calculus.

\begin{proposition}[Semantic Preservation of Spider Fusion] \label{prop:Spider Fusion}
Let $G \Rightarrow_{fuse} G_{fuse}$ denote the application of the spider fusion rule in ZX-calculus. 
Then
\[
\llbracket G \rrbracket = \llbracket G_{fuse} \rrbracket .
\]
\end{proposition}
\begin{proof}
See Appendix~\ref{subsec:proof-spider-fusion}.
\end{proof}

\subsubsection{\textbf{Layer Slicing of the ZX Diagram}}\label{Sec:OptimizetheInstructionOrderThroughTopologicalConnection}

After simplification, the next pass organizes the ZX nodes into layers before the backend constructs any pipe paths. We call this ZX-level organization \emph{layer slicing}, not scheduling: slicing annotates the IR with a realization-friendly frontier structure, whereas scheduling and path routing happen later when cubes and pipes are embedded in space--time. An effective slicing strategy exposes groups of ZX nodes whose incident paths can be realized near the same time step during backend construction.

\textbf{BFS Slicing Algorithm.} For this slicing process, we apply a breadth-first search on the ZX diagram. As shown in Figure~\ref{fig:ZX_Opt}(d), we start from the spiders adjacent to the input boundary wires, i.e., the first spider on each qubit wire, and label them as layer 1; for example, the first spider on $q_3$ is the target of a CNOT at circuit depth 2, but it is still in layer 1. Firstly, we visit all their unvisited immediate neighbors and label them as layer 2. For each of these neighbors, we then visit their unvisited neighbors and label them as layer 3. This process continues until all spiders have been visited and assigned a layer. Formally, each spider $v$ is assigned a layer index $L(v)$ via breadth-first traversal, while the ZX graph itself is left unchanged.

The key correctness criterion for this pass is not a change in the ZX semantics, but the validity of the layer-order annotation that will guide backend realization.
Let $G=(V,E)$ be the ZX graph after spider fusion, and let $L:V\rightarrow \mathbb{N}$ be a layer assignment. We call $L$ slicing-valid if every spider is assigned exactly one layer and every edge $(u,v)\in E$ satisfies
\[
|L(u)-L(v)|\leq 1.
\]
This property is what makes the annotation usable by the backend: every ZX edge is classified either as an inner connection within one layer or as an input/output connection between two adjacent layers. 
A secondary objective is to reduce the number of layers, which gives the embedding stage more opportunity to realize independent cubes near the same time step. The final space--time objective is not part of ZX slicing; it is optimized later by MCTS over pipe-level placements and pipe routes.
Slicing validity alone does not determine the BFS labels; for instance, assigning all spiders to a single layer is also slicing-valid. The backend additionally relies on a property that BFS guarantees by construction (Theorem~\ref{theorem:BFS Layer Slicing}): every spider in a layer $k>1$ has a neighbor in layer $k-1$, namely its BFS parent. When the backend embeds a spider, it generates candidate positions next to the already-embedded cube of such a neighbor (Section~\ref{Sec:MCTSfor3DEmbedding}), so this property gives every spider except those of the initial frontier, which connect to the input boundary, an anchor in the part of the diagram that is already realized.

\begin{theorem}[Validity of BFS Layer Slicing] \label{theorem:BFS Layer Slicing}
Let $L(v)$ be the BFS distance of a ZX spider $v$ from the input boundary wires, so that the first spider on each wire has $L(v)=1$. Then $L$ is slicing-valid: every ZX edge is either internal to one layer or crosses between two adjacent layers. Every spider $v$ with $L(v)>1$ has a neighbor $u$ with $L(u)=L(v)-1$. Moreover, BFS slicing only annotates the nodes of $G$ with layer labels; it does not add, delete, or rewrite any ZX node or edge, so the annotated IR has the same denotation as $G$.
\end{theorem}
\begin{proof}
See Appendix~\ref{subsec:proof-bfs-layer-slicing}.
\end{proof}

In this example (Figure~\ref{fig:ZX_Opt}(d)), the sliced IR has 4 layers, compared with the original circuit depth of 6.
Single-wire phase spiders are not labeled in the slicing figure; in the pipe diagram, they are realized by the one-pipe green or purple cubes summarized in Figure~\ref{fig:Gate}.

\textbf{Connection Types in Layers.} 
At this stage, the layer labels are still an IR-level organization: they specify which ZX edges become current realization obligations, but not the concrete geometric paths that will implement them. To pass useful information to MCTS, the compiler classifies the incident edges of each layer according to when their pipe paths must be realized.
Using the third layer from the previous example as a case study, we illustrate this in Figure~\ref{fig:ZX_Opt}(e). For a spider in layer 3, we define \textbf{input} as the edges coming from the previous layer (layer 2), shown with blue arrows; \textbf{output} as the edges going to the next layer (layer 4), shown with red arrows; and \textbf{inner connection} as the edges linking spiders within the same layer (layer 3), shown with yellow edges.

\begin{example}[Parallelism from Layer Slicing]\label{ex:ladder}
Before moving on to the next section, we present a small example to illustrate how breadth-first-search layer slicing exposes parallelism in lattice surgery programs. Figure~\ref{fig:Ladder} depicts a ladder circuit with 5 qubits. Executing the circuit gate by gate would require 5 steps to complete. However, if we slice the corresponding ZX spiders into layers and then realize those layers as pipe-diagram cubes, the process can be completed in just 2 steps, regardless of the number of qubits.
This effect is directly illustrated in the GHZ state generation shown in Figure~\ref{fig:ghz_16} of the evaluation section.
\end{example}

\subsection{Pipe-Level Realization and Optimization} \label{Sec:3DLayoutOptimization}

This subsection describes the backend realization pass required by Definition~\ref{def:PipeRealization}. Given an optimized ZX graph enriched with layer-slicing information, the backend instantiates ZX nodes as pipe-diagram cubes, realizes ZX edges as routed pipe paths, and uses MCTS to search for a valid low-volume space--time embedding. We refer to this task as the pipe construction problem.

We first define how each ZX node is instantiated as a corresponding pipe-diagram cube, and then describe how these cubes and their connecting pipes are assembled in 3D space using MCTS.

\subsubsection{ZX-to-Pipe Instantiation and Construction} \label{Sec:ZX-to-Pipe Instantiation}\label{Sec:ZXToPipeInstantiation}

\begin{wrapfigure}{R}{0.3\textwidth}
    \centering
    % \vspace{-10pt}
    \includegraphics[width=0.3\textwidth]{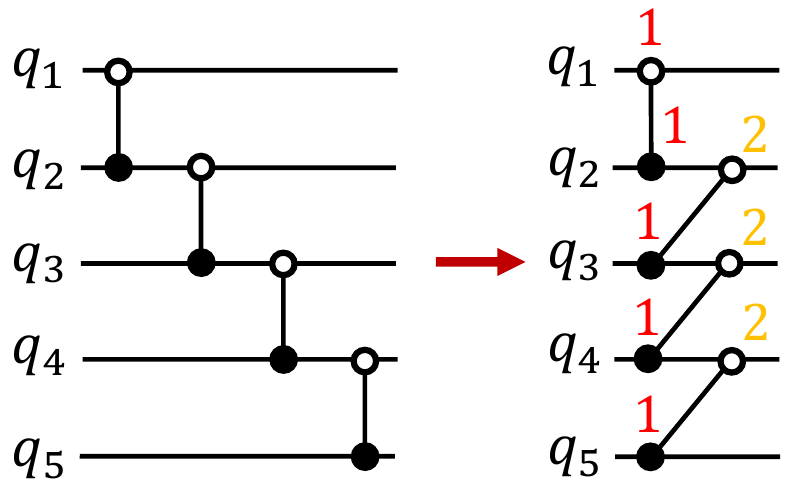}
    % \vspace{-10pt}
    \caption{Ladder circuit example}
    \vspace{-10pt}
    \label{fig:Ladder}
\end{wrapfigure}

To convert the optimized ZX diagram into a pipe diagram, each ZX node must be instantiated as its corresponding pipe-diagram cube.
Figure~\ref{fig:Gate} summarizes the cube types used for this instantiation by the construction procedure $\mathcal{I}$ of this section, which maps a ZX graph to a pipe diagram, including blue cubes, red cubes, yellow cubes, green cubes, and purple cubes; a detailed derivation of these correspondences is provided in Appendix~\ref{Sec:Gate}.
These correspondence rules implement the pipe-to-IR interpretation $\mathcal{T}_P$ from Definition~\ref{def:PipeToIR}: cubes realize ZX nodes, and pipes realize ZX edges.

The relationship between ZX diagrams and pipe diagrams can therefore be summarized as:

\[
G
\;\xrightarrow{\;\mathcal{I}\;}
P
\;\xrightarrow{\;\mathcal{T}_P\;}
G',
\]

where $G'$ is semantically equivalent to $G$. In the common case $G'=G$ syntactically, but routing may insert degree-two identity spiders or split a spider into connected same-color spiders, so the equality we rely on is $\llbracket G'\rrbracket=\llbracket G\rrbracket$.

For each ZX node, $\mathcal{I}$ first instantiates it as the corresponding pipe-diagram cube according to the correspondences illustrated in Figure~\ref{fig:Gate}, and then determines a valid geometric realization by assigning cube positions, cube orientations when required, and pipe routes between cubes that satisfy conditions (V1)--(V5) of Definition~\ref{def:ValidPipeDiagram}.
Determining such a geometric realization constitutes the core \textbf{embedding problem}, which will be addressed in Section~\ref{Sec:MCTSfor3DEmbedding} using MCTS.

\subsubsection{Embedding Strategy at High Level} \label{Sec:EmbeddingStrategyatHigh Level}

% We first provide a high-level overview of our method for embedding the ZX diagram into 3D space (two spatial axes and one temporal axis), as shown in Figure~\ref{fig:MCTS_H}. 

After defining ZX-to-pipe instantiation, we now describe how the backend chooses a concrete geometric realization: where instantiated cubes are placed and which valid pipe paths connect them. Figure~\ref{fig:MCTS_H} gives the high-level view of this embedding process in 3D space--time.

% \begin{figure}[t]
%     \centering
%     \includegraphics[width=0.6\linewidth]{fig/MCTS_H.pdf}
%     % \vspace{-5pt}
%     \caption{High-level overview of MCTS embedding}
%     % \vspace{-5pt}
%     \label{fig:MCTS_H}
% \end{figure}

\textbf{Layer-Wise Embedding Strategy.}
The embedding proceeds layer by layer following the IR organization produced by the slicing process described in Section~\ref{Sec:OptimizetheInstructionOrderThroughTopologicalConnection}.

\begin{wrapfigure}{R}{0.58\textwidth}
    \centering
    \includegraphics[width=0.58\textwidth]{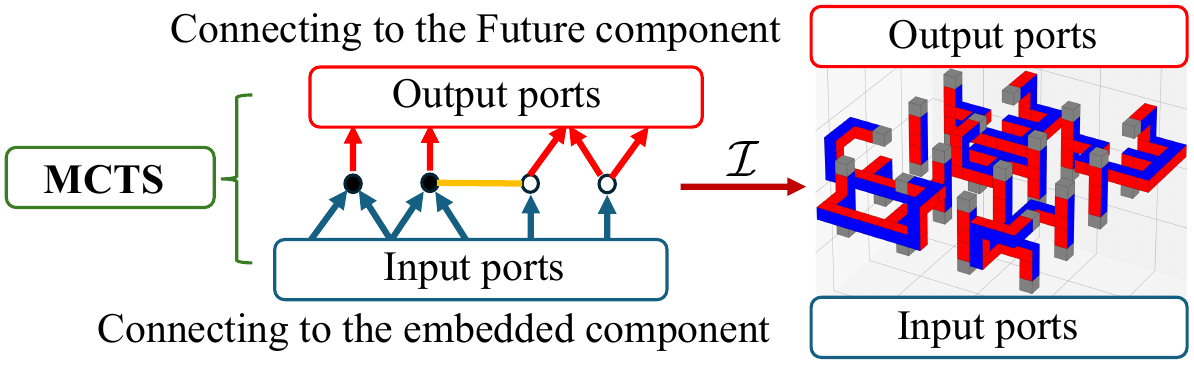}
    \caption{High-level overview of MCTS embedding}
    \label{fig:MCTS_H}
\end{wrapfigure}

Before embedding a given layer (illustrated in the middle of the left panel of Figure~\ref{fig:MCTS_H}), all spiders in the preceding layers have already been constructed as pipe diagram fragments. Their topological connections to spiders in the current layer are represented by blue arrows, which we refer to as \emph{input connections} (as defined in Section~\ref{Sec:OptimizetheInstructionOrderThroughTopologicalConnection}).

\textbf{\textbf{Inner-Layer Embedding and Connection Realization.}}
At this stage, MCTS determines the geometric realization of the current layer: it assigns positions and orientations to instantiated cubes and realizes the required ZX edges as valid pipe paths. These path-realization obligations include the input connections (blue arrows) and the inner connections within the layer (yellow edges).

Within each layer, spiders are embedded sequentially rather than simultaneously. A random ordering is assigned to the spiders in the layer, and MCTS incrementally embeds them one by one following this order. This ordering affects the search process but does not change the semantics of the resulting embedding.

\begin{wrapfigure}{R}{0.56\textwidth}
    \centering
    \includegraphics[width=0.56\textwidth]{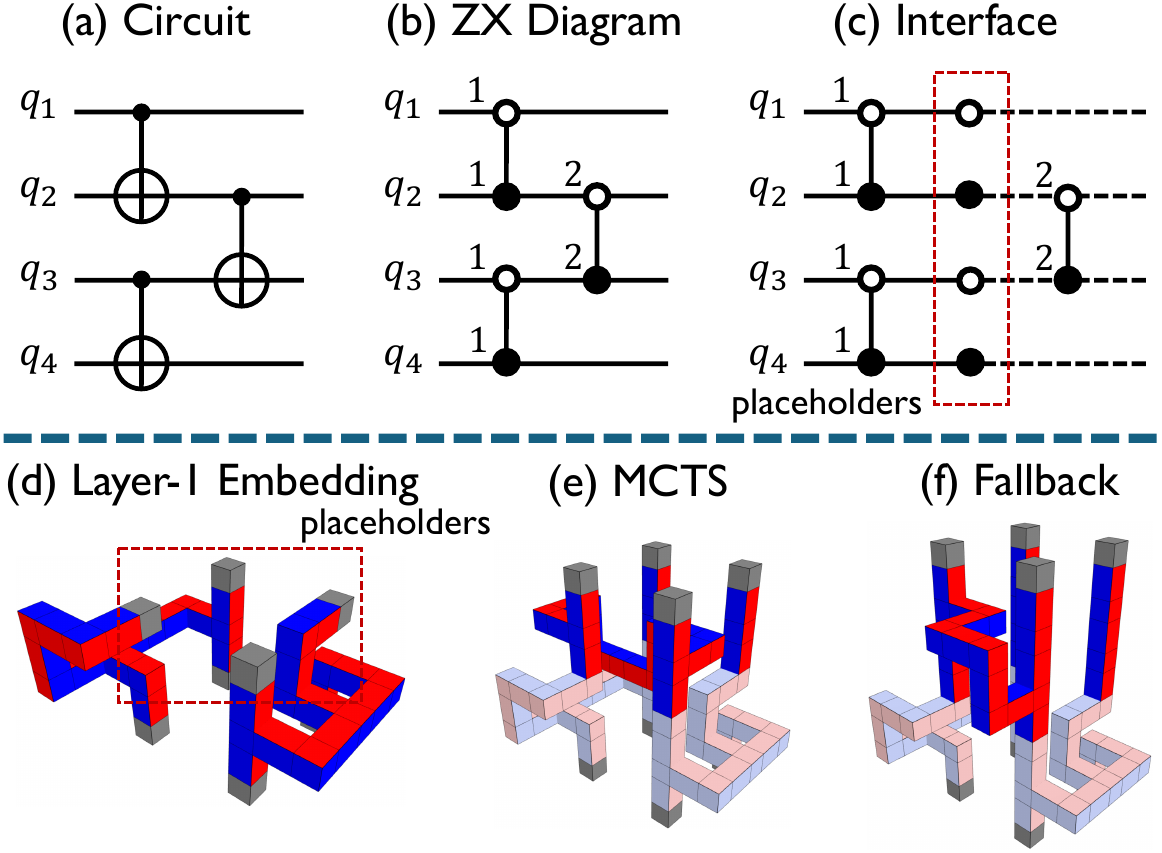}
    \caption{Placeholder interface of a 4-qubit circuit (layers 1 and 2 in separate blocks); layer~2 is embedded by (e)~MCTS or (f)~the fallback.}
    \label{fig:invariant}
\end{wrapfigure}

\textbf{Handling Future Connections.}
The output connections (red arrows), which link to spiders in future layers, cannot yet be realized as complete paths because their future endpoint cubes have not been instantiated. Instead, the backend materializes each such deferred path as a temporary placeholder port at the top boundary of the current embedding, as illustrated by the grey cubes in the right panel of Figure~\ref{fig:MCTS_H}. These ports record the unresolved realization obligations that will be routed when the corresponding future spiders are embedded.

Formally, if a spider $v$ in the current layer has one or more edges to future layers, we split $v$ into two same-color spiders connected by an identity edge: one spider keeps the already realized incident edges, while the other carries the unresolved future edges and is represented by the placeholder port. This is the inverse of spider fusion, so it preserves the ZX denotation. The implementation uses one placeholder per such spider, rather than one placeholder per outgoing edge, which bounds the routing frontier while preserving the arity of the deferred connections at the placeholder.

\textbf{Invariants Connecting the Layers.} After layer $k$ has been embedded, the partial pipe diagram satisfies three invariants. \textbf{(I1) Validity:} all cubes and pipes placed so far satisfy (V1)--(V5) (Proposition~\ref{prop:MCTS Validity Invariant}). \textbf{(I2) Semantics:} translating the realized part back to a ZX diagram and fusing each placeholder with its realized copy gives the ZX graph of layers $1,\ldots,k$, with every edge to a later layer left open; since each placeholder is a spider defusion of its frontier spider, the denotation is preserved (Proposition~\ref{prop:Semantic Preservation of ZX-to-Pipe Instantiation}). \textbf{(I3) Interface:} the placeholders occupy distinct free cells of the plane one unit above the highest occupied cell, each is connected to its frontier spider by a route disjoint from all other routes and pipes, and nothing lies above this plane (Figure~\ref{fig:invariant}(c,d)); the search rejects states whose frontier cannot be routed in this way.

Figure~\ref{fig:invariant} illustrates these invariants. After layer~1 is embedded, each of its spiders is split into a realized copy and a placeholder of the same type (Figure~\ref{fig:invariant}(c)), and the placeholders form the top of the embedding with nothing above them (Figure~\ref{fig:invariant}(d)). Each spider of the next layer has a neighbor in the current layer (Section~\ref{Sec:OptimizetheInstructionOrderThroughTopologicalConnection}), so its deferred edges start at placeholders, and the next layer is embedded from the placeholders into the empty region above them, without modifying or crossing the realized part. MCTS searches this region for a compact layout (Figure~\ref{fig:invariant}(e)). If the search fails within a block, the compiler falls back to a brute-force construction that re-embeds the block from its start, gate by gate (Figure~\ref{fig:invariant}(f)); this is possible because a block is a contiguous segment of the input circuit (Section~\ref{Sec:Scalability}). The construction extends each qubit by a vertical column above its placeholder and connects each two-qubit gate between two columns in a new time step. Since the region above the interface is empty, these connections can always be routed, so the fallback always succeeds, and its top again forms an interface satisfying (I1)--(I3).

The semantic correctness of the construction procedure, including the introduction of placeholder ports, is formalized in the following proposition.

\begin{proposition}[Semantic Preservation of ZX-to-Pipe Construction] \label{prop:Semantic Preservation of ZX-to-Pipe Instantiation}
Let
\[
P = \mathcal{I}(G)
\]
be the pipe diagram constructed from the ZX diagram $G$ using the construction procedure after the placeholder defusions described above.
Then
\[
\llbracket \mathcal{T}_P(P) \rrbracket
=
\llbracket G \rrbracket .
\]
\end{proposition}
\begin{proof}
See Appendix~\ref{subsec:proof-zx-to-pipe-construction}.
\end{proof}

% \begin{proposition}[Semantic Preservation under Placeholder Port Insertion]
% Let 
% \[
% P=\mathcal{I}'(G)
% \]
% be the pipe diagram constructed from the ZX diagram $G$ using the modified instantiation procedure with placeholder ports. Then
% \[
% \llbracket \mathcal{T}_P(P) \rrbracket
% =
% \llbracket G \rrbracket .
% \]
% \end{proposition}

% \begin{proof}
%     Introducing a placeholder port corresponds to applying a ZX spider defusion, where the original spider is split into two connected spiders: one remains at the current location while the other serves as a routing junction for future connections.
% \end{proof}

% Conceptually, this corresponds to applying a ZX spider defusion, where the original spider is split into two connected spiders: one remains at the current location, while the other acts as a routing junction for future connections. Since spider fusion/defusion preserves ZX semantics, this transformation does not change the represented computation. Therefore, although the ZX graph obtained from the pipe diagram may differ structurally from $G_{\mathrm{fuse}}$, their semantics remain identical:
% \[
% \llbracket \mathcal{T}_P(P) \rrbracket
% =
% \llbracket G_{\mathrm{fuse}} \rrbracket .
% \]

To maintain scalability, we further restrict the number of spiders with outgoing connections in each layer to be at most the number of logical qubits, which bounds the number of placeholder ports and prevents the embedding space from growing unboundedly. This assumption will be justified in Section~\ref{Sec:Scalability}, where we show that circuit partitioning ensures this condition holds.

\subsubsection{MCTS for Optimized 3D Embedding} \label{Sec:MCTSfor3DEmbedding}

Layout optimization is achieved by using the space--time volume of the resulting pipe diagram as the objective function of MCTS, guiding the search toward embeddings that reduce this volume. We now describe our implementation of MCTS for solving the ZX diagram embedding problem.

\textbf{Embedding State.}
We first define the embedding state in MCTS. An embedding state corresponds to a partial construction of the pipe diagram. 
It records the current geometric realization of the ZX spiders that have already been embedded and the routed connections constructed so far.

Formally, an embedding state maintains a partial assignment to the pipe-diagram variables in Definition~\ref{def:PipeDiagramRepresentation}. 
In particular, it contains the following components:

(1) \textbf{Cube Variables.} for each cube $v$, the state records its type $\tau_v$, spatial position $p_v$, and orientation $o_v$ when that orientation is defined.

(2) \textbf{Pipe Variables.} for each constructed pipe segment $e$, the state records its endpoints, direction $d_e$, and color-orientation variables $b_e,r_e$.

(3) \textbf{ZX-to-Cube Correspondence.}
the state maintains a mapping:
\[
\phi : P \rightarrow G
\]
that associates each cube in the current pipe diagram with the ZX spider it realizes.
Since not all ZX spiders are embedded at a given state, this mapping is generally not surjective.
Therefore, the pipe diagram $P$ stored in the state represents a partial embedding of the ZX diagram $G$.
This mapping enables the algorithm to determine the placement and topological connections of the corresponding cube for the next ZX spider during the embedding process.

Now we introduce the four phases of MCTS below.

% \textbf{Embedding State} serves as the tree node in MCTS and is defined by the following primitives: (1) \textbf{Position}: a dictionary storing the 3D positions of all embedded spiders; (2) \textbf{Spider Type}: a dictionary indicating whether a spider is an X-spider, a Z-spider (optionally with a phase), or a Hadamard box; (3) \textbf{Orientation}: a dictionary specifying the face direction of the pipe junction for the embedded spider, which is perpendicular to all its ports at the junction. As illustrated by the green arrow in Figure~\ref{fig:Embedding}(a), no path may connect to the spider from this face direction. This serves as a strict constraint on the pipe structure, as described in Section~\ref{Sec:InterpretingLatticeSurgeryUsingZX-calculus}. (4) \textbf{Path}: a list of paths connecting the spiders.

\textbf{Selection.}
Selection follows the UCT policy. Starting from the root state, the algorithm repeatedly selects the child node with the highest UCT score until a node eligible for expansion is reached. The UCT score of a child node $i$ is defined as

\[
\text{UCT}(i) =
\frac{R_i}{N_i}
+
c \sqrt{\frac{\ln N_p}{N_i}},
\]

where $R_i$ is the accumulated reward of node $i$, $N_i$ is the visit count of node $i$, $N_p$ is the visit count of its parent node, and $c$ is the exploration constant. The reward is defined as the negative space--time volume of the resulting pipe diagram, i.e., $R = -V(P)$, so that embeddings with smaller volume receive higher rewards.

\textbf{Expansion.}
Expansion performs a one-step transition between embedding states by embedding the cube instantiated by the next ZX spider. 
First, a non-embedded spider from the current layer is selected and instantiated as its corresponding cube according to the ZX-to-pipe instantiation rules described in Section~\ref{Sec:ZX-to-Pipe Instantiation}.

The cube is then placed at a vacant position in the 3D space. 
Candidate positions are generated from the neighboring locations of its already embedded neighbors in the ZX diagram. In the simplest strategy, the spider is placed directly above one of its neighbors, which provides a deterministic placement and accelerates the search. 
To better exploit the 3D embedding space, we optionally allow placements in any neighboring direction, generating multiple candidate positions. 
We refer to this configurable candidate generation strategy as \textbf{placement optimization}.

After determining the position of the new cube, routed connections are constructed to connect it with its embedded neighbors. 
Each such ZX edge is realized by choosing a valid path: a sequence of pipes and, when needed, intermediate identity cubes in the 3D pipe diagram.
Figure~\ref{fig:Next} illustrates an example where a third spider is instantiated, placed, and connected to two previously embedded spiders.

\begin{figure}[t]
    \centering
    \includegraphics[width=0.7\linewidth]{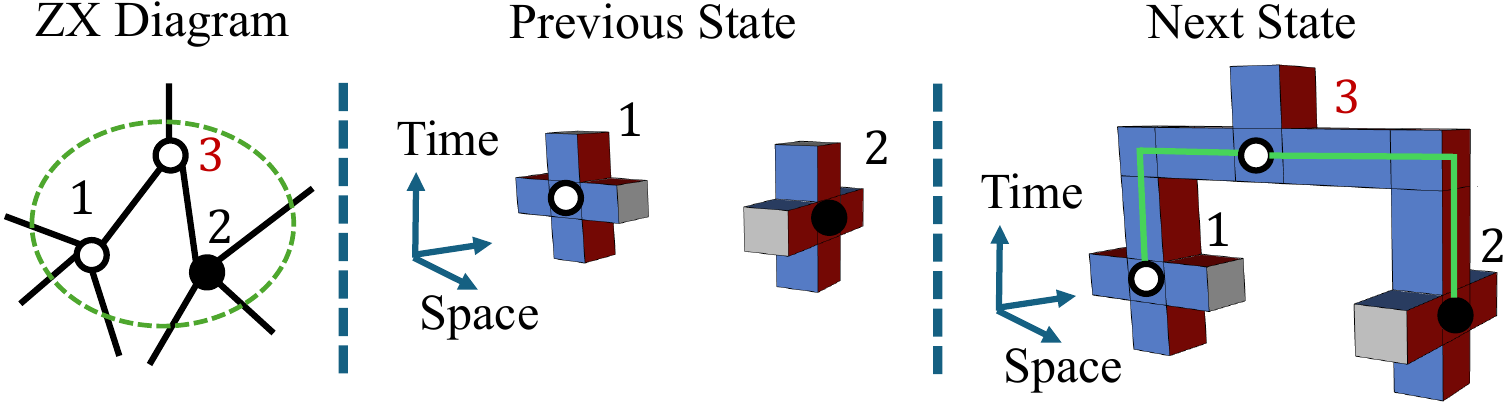}
    % \vspace{-5pt}
    \caption{Expansion example. The dangling wires of the ZX diagram connect the fragment inside the dashed ellipse to the rest of the graph; within a fragment they are not distinguished as inputs or outputs, which only the whole program designates.}
    % \vspace{-5pt}
    \label{fig:Next}
\end{figure}

We now discuss how the \textbf{validity conditions} (V1)--(V5) from Definition~\ref{def:ValidPipeDiagram} are enforced during routing. 
The routing process respects the blue/red cube condition (V3) by avoiding routed connections that traverse a cube along its orientation direction.
The shortest feasible path may violate the color-orientation part of (V3)--(V4). 
When such a conflict occurs, a color-switch operation is applied, which introduces additional color-switching cubes to adjust the pipe color orientation, as illustrated in Figure~\ref{fig:Color}. The additional blue or red cubes inserted along the route are again two-pipe routing cubes, and therefore correspond under $\mathcal{T}_P$ to degree-two zero-phase identity spiders.

If no feasible placement or routing satisfying conditions (V1)--(V5) can be found for the selected spider, the expansion fails and the state is treated as an infeasible embedding. In this case, the subsequent simulation assigns a reward of $-\infty$, effectively pruning this branch of the search.

\begin{proposition}[Validity Invariant of Pipe Construction] \label{prop:MCTS Validity Invariant}
Every accepted embedding state maintained by \myCompilerNameSpace satisfies the applicable validity conditions (V1)--(V5) of Definition~\ref{def:ValidPipeDiagram} for all cubes and pipe segments already constructed. If the MCTS search cannot extend such a state to a complete valid diagram for a layer, the compiler rebuilds the current block from its input interface with a canonical gate-by-gate fallback construction, whose boundary pipes match the already realized neighboring interfaces. Consequently, every completed layer emitted by the backend is a valid pipe diagram fragment.
\end{proposition}
\begin{proof}
See Appendix~\ref{subsec:proof-mcts-validity-invariant}.
\end{proof}

\textbf{Simulation.}
Starting from the expanded state, the simulation phase completes the embedding of the remaining spiders in the current layer using a fast heuristic rollout policy rather than further MCTS exploration. 
In this rollout, each remaining spider is embedded using a fast deterministic placement and routing rule that maintains conditions (V1)--(V5). The heuristic prioritizes finding a valid completion, even if the resulting embedding has a larger space--time volume than the best possible one. If such a valid completion is found, the reward is the negative space--time volume of the completed pipe diagram; otherwise the reward is $-\infty$.

\begin{wrapfigure}{R}{0.3\textwidth}
    \centering
    \includegraphics[width=0.3\textwidth]{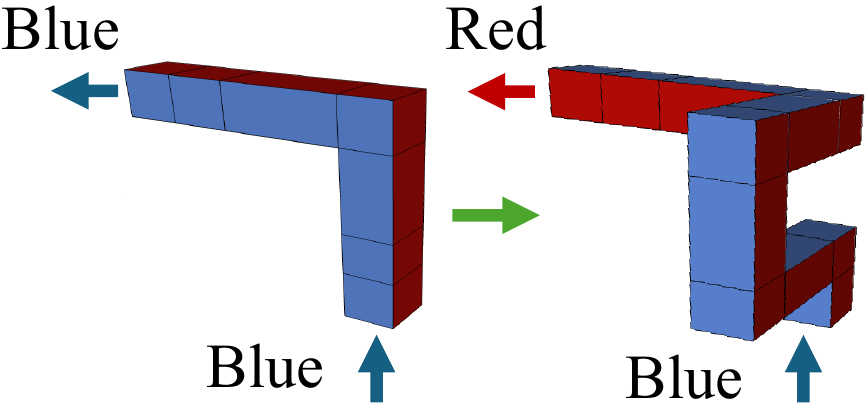}
    % \vspace{-5pt}
    \caption{Color switch}
    % \vspace{-5pt}
    \label{fig:Color}
\end{wrapfigure}

\textbf{Backpropagation.}
The obtained reward is propagated along the path from the simulated node to the root of the MCTS tree. 
For each visited node, the cumulative reward and visit count are updated accordingly.

\textbf{Applying MCTS to Layer-Wise Embedding.}
The MCTS procedure described above is applied to each layer of the ZX diagram in the order determined by the slicing process, as described in Section~\ref{Sec:EmbeddingStrategyatHigh Level}. 
For a given layer, MCTS searches for an embedding of all spiders in that layer while respecting the previously embedded layers. 
Once a layer is successfully embedded, its placement becomes fixed, unless the fallback rebuilds its block, and serves as the foundation for embedding subsequent layers.
In the first layer, all qubits are initially placed at grid positions. 
If the embedding of a layer fails, the algorithm rebuilds the current block from its start with the baseline compilation method (Section~\ref{Sec:EmbeddingStrategyatHigh Level}), a canonical gate-by-gate construction whose boundary pipes are matched to the surrounding realized interfaces. This fallback is important for correctness: MCTS affects the quality of the layout that is found, while validity and semantic preservation are enforced by the accepted-transition checks and by the fallback construction.

\subsection{Correctness of \myCompilerNameSpace}

We now state the correctness of the \myCompilerNameSpace compilation algorithm defined in this section.

\begin{theorem}[\myCompilerNameSpace Satisfies the IR Realization Contract] \label{theorem:IR Realization Contract}
Let $C$ be an input quantum circuit and let $P$ be the pipe diagram compiled by \myCompilerName\ from $C$. Define
\[
G_C = \mathcal{T}_C(C), \qquad
G_P = \mathcal{T}_P(P),
\]
using the translations from Definitions~\ref{def:CircuitToIR} and~\ref{def:PipeToIR}.
Then
\[
\llbracket G_C \rrbracket = \llbracket G_P \rrbracket
\qquad\text{and}\qquad
P\in\mathcal{P}_{\mathrm{valid}}.
\]
Equivalently, $P$ realizes $G_C$ in the sense of Definition~\ref{def:PipeRealization}.
\end{theorem}
\begin{proof}
See Appendix~\ref{subsec:proof-ir-realization-contract}.
\end{proof}

\section{Topology-Aware IR Partitioning for Scalability} \label{Sec:Scalability}

The previous section describes how \myCompilerNameSpace realizes a sliced ZX IR as a valid pipe diagram. This section adds a scalability pass to that pipeline: topology-aware IR partitioning. Like the earlier ZX-level passes, partitioning operates before backend realization and preserves the represented computation; it augments layer slicing to control the size of the pipe-realization problem exposed to MCTS.

In the terminology of Principle~\ref{prin:IRPreservingCompilation}, partitioning is an IR-preserving analysis pass: it chooses block boundaries before ZX slicing, inserts identity spiders only to reconnect consecutive blocks, and then sends each block to the same backend realization procedure. Thus the pass changes only the granularity at which the IR is realized as pipe fragments; it does not change the ZX denotation.

\begin{wrapfigure}{R}{0.58\textwidth}
    \centering
    \includegraphics[width=0.58\textwidth]{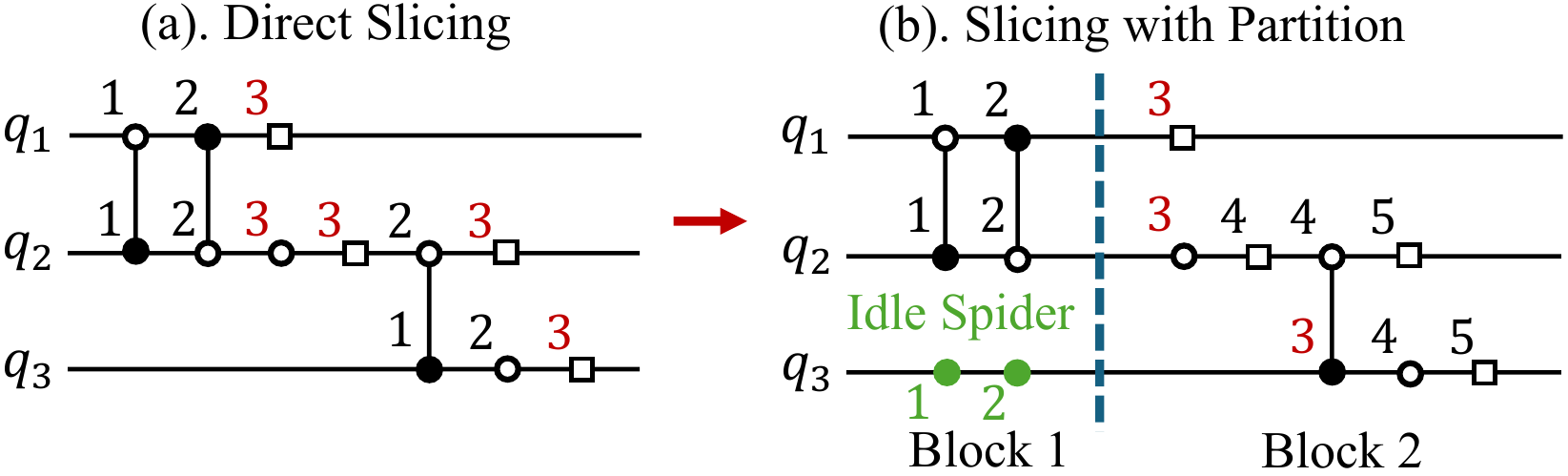}
    % \vspace{-5pt}
    \caption{Circuit partition example}
    \vspace{-5pt}
    \label{fig:Partition}
\end{wrapfigure}

\textbf{Routing Frontier Explosion.}
The layer slicing pass introduced in Section~\ref{Sec:OptimizetheInstructionOrderThroughTopologicalConnection} can significantly parallelize spider operations, but it may also expose many deferred paths at the same layer boundary. These paths determine the \textbf{routing frontier}: the set of current-layer spiders whose outgoing ZX edges must be carried forward as placeholder ports before the backend can complete their pipe realization.
The number of frontier spiders directly determines the number of placeholder ports introduced for future path realization, as shown in Figure~\ref{fig:MCTS_H}, and therefore determines the spatial footprint required for embedding the layer.
If this number varies significantly across layers, some layers may require a much larger embedding space than others, leading to inefficient use of space--time resources.

\begin{example}[Routing Explosion Across Layers]

An example is shown in Figure~\ref{fig:Partition}(a), where a 3-qubit diagram contains 4 spiders with outgoing connections in the second layer and requires embedding 5 spiders in the third layer, which is almost twice the number of qubits. As the number of qubits increases, the number of spiders with outgoing connections may also grow proportionally. This increases the number of placeholder ports introduced for future connections, enlarging the routing frontier and the spatial footprint required for embedding.
\end{example}

\textbf{IR Partitioning.}
A natural way to control this backend realization pressure is to partition the IR into smaller blocks before each block is sliced and embedded.
As shown in Figure~\ref{fig:Partition}(b), introducing a single partition in the middle reduces the number of spiders with outgoing connections in the second layer from 4 to 3, which matches the number of logical qubits. At the same time, the number of spiders in the third layer is reduced from 5 to 3, making the structure easier to embed. However, this also increases the total number of layers from 3 to 5. In the extreme case where the ZX graph IR is partitioned at every depth, the maximum number of spiders with outgoing connections per layer is equal to the number of qubits, while the total number of layers equals the circuit depth.

\textbf{Maintaining Layer Consistency.}
After partitioning, the circuit is divided into separate blocks, which requires slight modifications to the slicing algorithm described in Section~\ref{Sec:OptimizetheInstructionOrderThroughTopologicalConnection}. For each new block, we assign spider labels starting from one plus the maximum layer index of its preceding blocks. For example, in Figure~\ref{fig:Partition}(b), block 2 begins at layer 3. However, partitioning may sometimes disrupt layer consistency, causing neighboring spiders to have non-consecutive layer numbers. To resolve this, we insert idle spiders to enforce consistency, as illustrated by the green node in Figure~\ref{fig:Partition}(b). These idle spiders correspond to identity operations in the ZX representation and therefore do not change the semantics of the ZX diagram.

\textbf{Limitations of Uniform Partitioning.}
A simple strategy is uniform partitioning, which splits the graph IR at fixed depths (e.g., every five circuit-depth levels) to control the routing frontier. However, this increases the total number of layers to approximately match the original circuit depth, thereby diminishing the benefits of ZX-based optimization.

\textbf{Topology-Aware Partitioning.}
Our pass instead chooses block boundaries from the topology that the later realization will actually see: the outgoing spiders produced by slicing. The depth of each partition block is therefore determined dynamically.
For each new block, the partitioning begins with a circuit depth of 2. We then apply slicing and examine the spiders with outgoing connections in each layer of the block.

Let $S_i$ denote the set of spiders in layer $i$ that have outgoing connections to the next layer. These spiders correspond to routing frontiers that must be connected during embedding. If $|S_i|$ remains at or below a predefined threshold for all layers in the block, the block is extended by one additional circuit depth. Otherwise, the partition is finalized at the preceding circuit depth.

In our implementation, the threshold is set to the number of logical qubits. Consequently, each block exposes at most that many deferred path obligations to the backend, bounding the number of placeholder ports and preventing the MCTS realization space from growing unboundedly. This threshold is a compiler parameter rather than a semantic assumption: a larger threshold preserves more cross-block topological freedom but gives MCTS a larger routing frontier, while a smaller threshold improves search scalability at the cost of more block boundaries.

Appendix~\ref{Sec:E2E} traces a 2-qubit circuit through all compilation stages.

\section{Evaluation} \label{Sec:Evaluation}

In this section, we evaluate \myCompilerNameSpace by comparing it against state-of-the-art baselines, measuring the contribution of each optimization, examining different hardware configurations, and assessing scalability. These experiments show how the IR transformations, backend embedding choices, and topology-aware partitioning affect the final space--time volume.

\subsection{Experiment Setup.} \label{Sec:ExperimentsSetup}

\textbf{Baseline.} Our baseline compilers include two heuristic compilers based on the quantum circuit model for lattice surgery compilation: Liblsqecc~\cite{watkins2024high} and DASCOT~\cite{molavi2025dependency}, as well as a SAT-solver-based compiler LaSsynth~\cite{tan2024sat} for Clifford circuits.
Note that Liblsqecc can manage the physical qubit resources in two configurations: `sparse layout', where data qubits are interleaved with ancilla qubits, and `tight layout', where all data qubits are arranged in two rows with a single row of ancilla qubits placed between them. \myCompilerName, DASCOT, and LaSsynth all adopt the `sparse layout'.

\textbf{Experiment Configurations.} To isolate where the improvement comes from in the compilation pipeline, we design three configurations. All three configurations incorporate the ZX-level passes in Section~\ref{Sec:TopologicalOptimizationat}, including spider fusion and layer slicing.
1. \textsf{`Place-Opt'} enables placement optimization during the expansion phase of MCTS (see Section~\ref{Sec:MCTSfor3DEmbedding}), allowing exploration of spider placement in 3D space.
2. \textsf{`Part-Opt'} applies the topology-aware circuit partitioning method introduced in Section~\ref{Sec:Scalability}, in contrast to a uniform partitioning applied every 5 circuit depths.
3. \textsf{`Full-Opt'} applies all \myCompilerNameSpace optimizations.
On top of these shared ZX-level passes in Section~\ref{Sec:TopologicalOptimizationat}, \textsf{Place-Opt} emphasizes the backend realization choice of where cubes and routes are placed, \textsf{Part-Opt} emphasizes the IR partitioning pass that controls routing frontiers, and \textsf{Full-Opt} evaluates their combined effect.

\textbf{Hardware Configuration.} 
We adopt a two-dimensional lattice coupling of physical qubits as our device architecture. The logical qubits are square patches on the physical qubit lattice. This is a common setting widely used in previous works~\cite{watkins2024high,molavi2025dependency,tan2024sat}.

\textbf{Benchmarks.} We evaluate \myCompilerNameSpace using a variety of quantum algorithms of different types and sizes, with detailed information provided in Table~\ref{tab:benchmark} (Appendix~\ref{Sec:Benchmarks}). Appendix~\ref{Sec:Benchmarks} describes the source and construction of each circuit.

\textbf{Metrics.} We mainly use the following two metrics:
 (\textbf{a}) \textbf{Space--time volume:} The product of space and time occupied by the lattice surgery program, representing the physical resources consumed. (\textbf{b}) \textbf{Compilation time:}  The CPU time required to translate a quantum circuit into executable lattice surgery instructions.

\textbf{Implementation.} We implemented \myCompilerNameSpace in Python, utilizing the PyZX package~\cite{kissinger2019pyzx} to help with ZX-diagram-based optimizations. We chose PyZX for its Python interface and support for ZX-diagram manipulation; our approach is not tied to this library, and other tools that can manipulate ZX diagrams will also work. The ZX-level transformations we use (spider fusion, its inverse, and identity insertion) can be expressed and verified in frameworks such as VyZX~\cite{lehmann2026vy} and Quantomatic~\cite{kissinger2015quantomatic}; we rely on the established soundness of these rules and have not performed a separate mechanized verification. All experiments were conducted on a server with a 2.4 GHz CPU and 0.75 TB of memory. For the MCTS implementation, we selected two random seeds for each layer embedding to change the embedding order of spiders, with the number of iterations fixed at 1000 and a timeout of 2 seconds. Among the results from different spider embedding orders, the one yielding the smaller space--time volume was chosen. These settings make \myCompilerNameSpace substantially more expensive than the fastest gate-centric baselines, but they target a different point in the design space: spending compilation time to reduce the space--time volume of the produced lattice-surgery program. 
Regarding the magic state, we take the $R_z$ state as the magic state and adopt the assumption from the baseline~\cite{molavi2025dependency}: the magic state factory is not included in the space--time volume, and the distillation time is omitted. We use this convention only for fair comparison with the baseline.
%. We also omit the distillation time

\subsection{Overall Result} \label{Sec:OverallResult}

\begin{wrapfigure}{R}{0.3\textwidth}
    \centering
    % \vspace{-10pt}
    \includegraphics[width=0.3\textwidth]{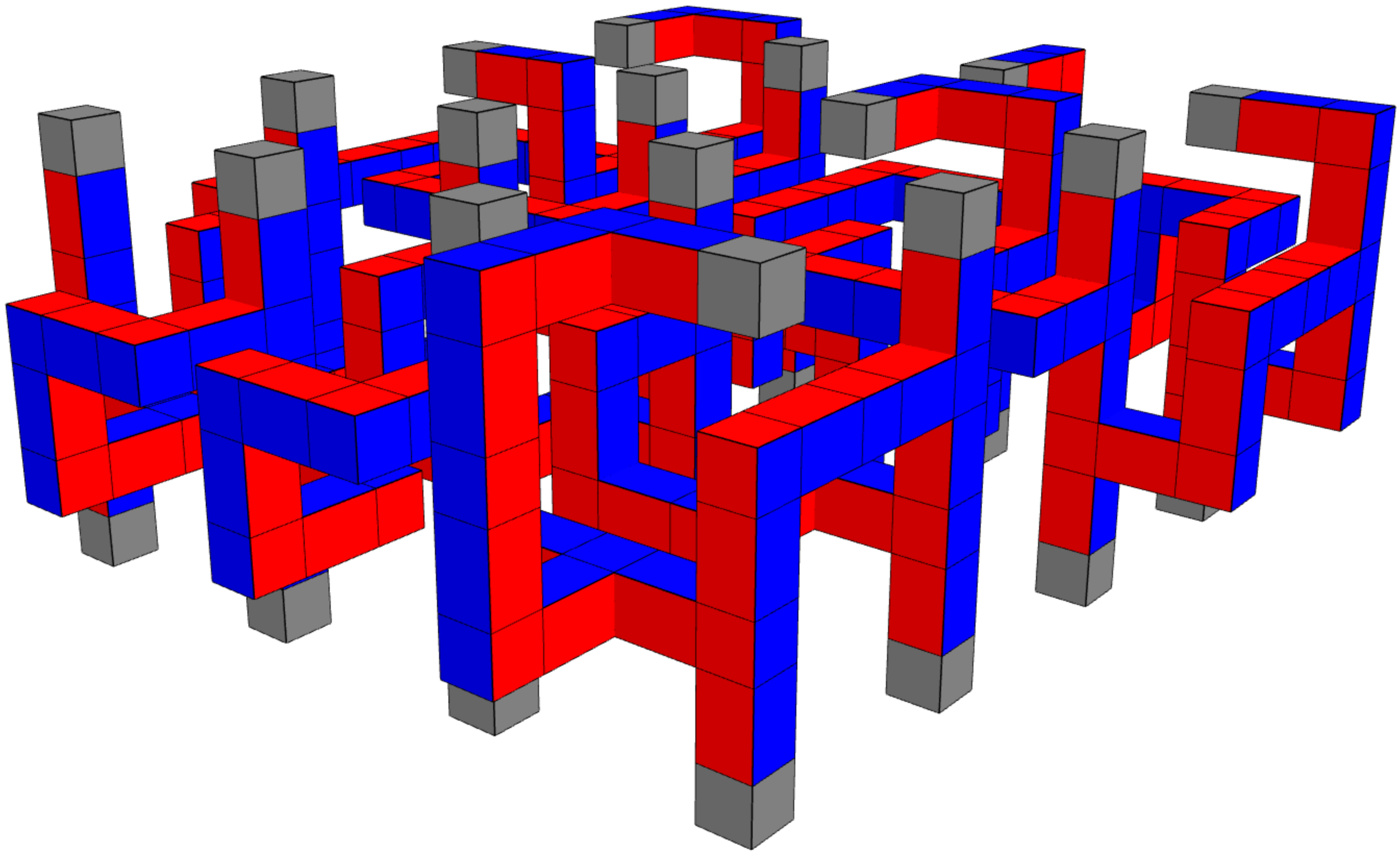}
    % \vspace{-10pt}
    \caption{Compiled pipe diagram for 16-qubit GHZ state}
    \vspace{-10pt}
    \label{fig:ghz_16}
\end{wrapfigure}

Table~\ref{tab:Overall} summarizes the overall compilation results of \myCompilerNameSpace for nine benchmark algorithms with 16 qubits (except for Grover search, which uses 6 qubits). The circuit depths range from $12$ to $457$, and the gate counts range from $16$ to $1211$. For the sparse layout configuration, we adopt a square lattice hardware topology in which each row and column contains 4 data qubits interleaved with ancilla qubits, forming a $4\times4$ grid. For Grover's search, the layout is a $2\times3$ grid.

% \begin{figure}[t]
%     \centering
%     \includegraphics[width=0.5\linewidth]{fig/ghz_16.png}
%     % \vspace{-5pt}
%     \caption{Compiled pipe diagram for 16-qubit GHZ state}
%     % \vspace{-5pt}
%     \label{fig:ghz_16}
% \end{figure}

Compared with the best space--time volume results among the baseline compilers, the best results from \myCompilerNameSpace achieve an average reduction of $46\%$. We provide the detailed reduction for each benchmark, listed alongside their names in the table. The SAT-solver-based compiler LaSsynth is not included in the table, as many of the benchmarks are non-Clifford circuits and the system sizes are too large for it to handle.
Across benchmarks, the reduction ranges from $25\%$ to $90\%$, reflecting their inherently different structures.
We show the pipe diagram for GHZ state generation in Figure~\ref{fig:ghz_16}, which spans 4 time steps, with 2 of them allocated to the input and output ports.
This is the most direct example of the pass-to-realization story: layer slicing exposes the ladder as a 2-layer topological structure, and the backend realizes those layers as the compact pipe diagram shown in Figure~\ref{fig:ghz_16}.

\begin{table}[t]
  \centering
%   \vspace{-5pt}
  \caption{Overall Comparison. `$^*$' indicates the best result among three configurations of \myCompilerName; `$^{\diamond}$' indicates the best result among three baseline configurations. Reduction in space--time volume is in the brackets next to the benchmark name.}

\begin{adjustbox}{width=\textwidth}
  {\fontsize{9}{12}\selectfont
    \begin{tabular}{|m{4cm}<{\centering}|c|c|c|c|c|c|}
    \hline
     Benchmark (Ave: $46\%$) & \multicolumn{2}{c|}{BV ($62\%$)} &\multicolumn{2}{c|}{DJ ($59\%$)}& \multicolumn{2}{c|}{Grover Search ($25\%$)}   \\
    \hline
    % Benchmark & Qubit\# & \multicolumn{2}{c|}{Circuit Metrics} \\
    % \hline
    Compiler $/ $ Metrics & Space--time & Comp time (s) & Space--time & Comp time (s) & Space--time & Comp time (s) \\ \hline
    \myCompilerNameSpace(\textsf{Place-Opt}) & 891 & 45.1 & 1215 & 72.6 & 23065$^*$ & 2253.4 \\ \hline
    \myCompilerNameSpace(\textsf{Part-Opt}) & 729 & 9.6 & 1053 & 19.7 & 29960 & 1126.8 \\ \hline
    \myCompilerNameSpace(\textsf{Full-Opt}) & 486$^*$ & 26.3 & 891$^*$ & 55.8 & 23240 & 2412.0 \\ \hline
    DASCOT & 2511 & 1.1 & 3969 & 1.1 & 33425 & 3.2 \\ \hline
    Liblsqecc (tight) & 1290$^{\diamond}$ & 0.05 & 2190$^{\diamond}$ & 0.08 & 30675$^{\diamond}$ & 1.8 \\ 
    \hline
    Liblsqecc (sparse) & 4257 & 0.08 & 7029 & 0.14 & 44415 & 1.4 \\ 
    \hline
     & \multicolumn{2}{c|}{QFT ($35\%$)} &\multicolumn{2}{c|}{QPE ($39\%$)}& \multicolumn{2}{c|}{VQE ($31\%$)}   \\
    \hline
    % Benchmark & Qubit\# & \multicolumn{2}{c|}{Circuit Metrics} \\
    % \hline
    Compiler $/ $ Metrics & Space--time & Comp time (s) & Space--time & Comp time (s) & Space--time & Comp time (s) \\ \hline
    \myCompilerNameSpace(\textsf{Place-Opt}) & 37098 & 2984.6 & 40419 & 3280.5 & 4779 & 378.1 \\ \hline
    \myCompilerNameSpace(\textsf{Part-Opt}) & 39690 & 1464.3 & 45603 & 1573.3 & 4617 & 230.5 \\ \hline
    \myCompilerNameSpace(\textsf{Full-Opt}) & 36531$^*$ & 3097.0 & 39447$^*$ & 3316.6 & 4212$^*$ & 510.8\\ \hline
    DASCOT & 56133$^{\diamond}$ & 675.9 & 64152$^{\diamond}$ & 821.2 & 6075$^{\diamond}$ & 15.9\\ \hline
    Liblsqecc (tight) & 93180 & 3.6 & 96930 & 3.8 & 20580 & 0.77 \\ 
    \hline
    Liblsqecc (sparse) & 175230 & 4.3 & 192357 & 4.5 & 42471 & 0.95\\ 
    \hline
     & \multicolumn{2}{c|}{GHZ State ($90\%$)} &\multicolumn{2}{c|}{W State ($40\%$)}& \multicolumn{2}{c|}{QAOA ($33\%$)}   \\
    \hline
    % Benchmark & Qubit\# & \multicolumn{2}{c|}{Circuit Metrics} \\
    % \hline
    Compiler $/ $ Metrics & Space--time & Comp time (s) & Space--time & Comp time (s) & Space--time & Comp time (s) \\ \hline
    \myCompilerNameSpace(\textsf{Place-Opt}) & 972 & 29.3 & 9639 & 703.5 & 4374 & 449.4 \\ \hline
    \myCompilerNameSpace(\textsf{Part-Opt}) & 486 & 13.2 & 11988 & 285.2 & 4293$^*$ & 208.2 \\ \hline
    \myCompilerNameSpace(\textsf{Full-Opt}) & 243$^*$ & 12.7 & 8505$^*$ & 686.2 & 4374 & 402.6 \\ \hline
    DASCOT & 2511$^{\diamond}$ & 1.2 & 14094$^{\diamond}$ & 39.2 & 6399$^{\diamond}$ & 6.0 \\ \hline
    Liblsqecc (tight) & 3000 & 0.11 & 14850 & 0.59 & 10320 & 0.39 \\ 
    \hline
    Liblsqecc (sparse) & 3267 & 0.07 & 37521 & 0.84 & 17424 & 0.43 \\ 
    \hline
    
    \end{tabular}%
    }
  \label{tab:Overall}%
  \end{adjustbox}
\end{table}%

In addition, Table~\ref{tab:Overall} separates two sources of improvement. Placement optimization changes the backend realization while preserving the same sliced IR, yielding a $35\%$ average reduction. Topology-aware partitioning changes the IR block structure seen by the backend, yielding a $34\%$ average reduction.

\subsection{Spider Fusion and Search Sensitivity} \label{Sec:Ablation}

\begin{table}[t]

\caption{Space--time volume and compilation time in seconds of Full-Opt when one search parameter is changed or spider fusion is disabled. Ref.: Full-Opt in Table~\ref{tab:Overall} (seeds 0\&1; 1000 iterations; 2\,s per MCTS call)}
\label{tab:sensitivity}
\footnotesize
\setlength{\tabcolsep}{2.6pt}
\begin{tabular}{l|l|r|rr|rrrrr|rr|r}
\hline
Benchmark & Metric & Ref. & Seeds 2\&3 & Seeds 4\&5 & 20 it. & 50 it. & 100 it. & 200 it. & 500 it. & 1.5\,s & 3\,s & No fusion \\
\hline
\multirow{2}{*}{BV} & Volume & 486 & 1053 & 810 & 486 & 486 & 486 & 486 & 486 & 486 & 486 & 405 \\
 & Time (s) & 26 & 86 & 87 & 19 & 19 & 21 & 24 & 27 & 23 & 41 & 20 \\
\hline
\multirow{2}{*}{DJ} & Volume & 891 & 1377 & 1458 & 891 & 891 & 891 & 891 & 891 & 1458 & 891 & 810 \\
 & Time (s) & 56 & 153 & 161 & 30 & 37 & 45 & 53 & 61 & 119 & 89 & 56 \\
\hline
\multirow{2}{*}{Grover} & Volume & 23240 & 22715 & 23030 & 24640 & 23170 & 23450 & 23240 & 23240 & 23240 & 23205 & 22680 \\
 & Time (s) & 2412 & 2919 & 2930 & 383 & 621 & 1030 & 1527 & 2388 & 2616 & 3452 & 2230 \\
\hline
\multirow{2}{*}{QFT} & Volume & 36531 & 37341 & 36774 & 38313 & 37179 & 37827 & 37341 & 36126 & 37503 & 36288 & 36531 \\
 & Time (s) & 3097 & 3214 & 3161 & 970 & 1620 & 2238 & 2607 & 2908 & 2491 & 4402 & 3178 \\
\hline
\multirow{2}{*}{QPE} & Volume & 39447 & 39285 & 40257 & 43335 & 39690 & 40014 & 40176 & 40338 & 39933 & 38637 & 39366 \\
 & Time (s) & 3317 & 3582 & 3568 & 1121 & 1671 & 2203 & 2751 & 3292 & 2850 & 4757 & 3339 \\
\hline
\multirow{2}{*}{VQE} & Volume & 4212 & 4374 & 4131 & 4455 & 4455 & 4374 & 4455 & 4293 & 4455 & 3969 & 4293 \\
 & Time (s) & 511 & 463 & 397 & 225 & 325 & 499 & 541 & 517 & 301 & 747 & 392 \\
\hline
\multirow{2}{*}{GHZ} & Volume & 243 & 243 & 324 & 243 & 243 & 243 & 243 & 243 & 243 & 243 & 243 \\
 & Time (s) & 13 & 15 & 24 & 8 & 8 & 9 & 11 & 13 & 11 & 21 & 14 \\
\hline
\multirow{2}{*}{W State} & Volume & 8505 & 8667 & 8424 & 9477 & 8829 & 8586 & 8505 & 8505 & 8505 & 8505 & 9963 \\
 & Time (s) & 686 & 837 & 823 & 157 & 249 & 447 & 613 & 739 & 687 & 1138 & 891 \\
\hline
\multirow{2}{*}{QAOA} & Volume & 4374 & 4536 & 4212 & 4212 & 4131 & 4131 & 4131 & 4131 & 4698 & 4131 & 4131 \\
 & Time (s) & 403 & 394 & 537 & 181 & 274 & 332 & 376 & 411 & 452 & 608 & 378 \\
\hline
\end{tabular}
\end{table}

\textbf{Spider fusion.} Compared with Full-Opt, the last column of Table~\ref{tab:sensitivity} disables spider fusion and keeps all other settings unchanged. Without fusion, the volume changes by $-17\%$ to $+17\%$ relative to Full-Opt, and the average reduction over the best baseline stays at 46\%. Most of the reduction therefore comes from the three-dimensional placement and routing of the backend; deciding when to fuse is future work.

\textbf{Search parameters.} The other columns change one search parameter at a time; all changes below are relative to the reference column. \emph{Seeds:} with other seeds, the volume changes by at most 4\% for six of the nine benchmarks; BV, DJ, and GHZ are more sensitive (up to $+117\%$, $+64\%$, and $+33\%$): their volumes are small, so a few extra time steps change them by a large fraction. \emph{Iterations:} with 20 instead of 1000 MCTS iterations per call, the volume changes by at most 11.4\%, and the compilation time of the four slowest benchmarks (Grover, QFT, QPE, and W State) drops by a factor of 3--6; with 50 or more iterations, every volume stays within 6\% of the reference. \emph{Time limit per call:} lowering the limit from 2\,s to 1.5\,s shortens the compilation of QFT, QPE, and VQE by 14--41\% and increases their volume by at most 6\%, but gives no consistent gain for the other benchmarks; raising it to 3\,s makes compilation 42--66\% slower and reduces the volume by at most 6\%.

\textbf{Fallback.} The fallback construction (Section~\ref{Sec:MCTSfor3DEmbedding}) is never needed in the experiments of Table~\ref{tab:Overall} or Table~\ref{tab:sensitivity}.

\subsection{Hardware Structure Impact} \label{Sec:HardwareStructureImpact}

We next test whether the backend realization stage remains effective when the available geometry changes, with the results summarized in Table~\ref{tab:Arch}.
We select QFT, QPE, VQE as benchmark circuits. For \myCompilerName, full optimization was applied, while for Liblsqecc only the sparse layout was considered, since its compact layout is fixed for a given qubit number and has already been examined in Section~\ref{Sec:OverallResult}. The hardware architectures evaluated include a $6 \times 3$ grid, an $8 \times 2$ grid, and a linear $16 \times 1$ grid. The average reductions achieved are $35\%$, $38\%$, and $38\%$, respectively, which are consistent across different architectures and comparable to the results reported in Section~\ref{Sec:OverallResult} for the $4 \times 4$ grid.
This suggests that the improvement is not tied to one fixed layout: the IR passes expose topological alternatives, and the backend adapts their concrete realization to the available grid.

\begin{table}[t]
  \centering
%   \vspace{-5pt}
  \caption{Compilation results under different hardware structures (volume reduction in brackets).}

\begin{adjustbox}{width=\textwidth}
  {\fontsize{9}{12}\selectfont
    \begin{tabular}{|m{4cm}<{\centering}|c|c|c|c|c|c|}
    \hline
    $\bm{6 \times 3}$ \textbf{grid} (Ave: $35\%$) & \multicolumn{2}{c|}{QFT ($31\%$)} &\multicolumn{2}{c|}{QPE ($39\%$)}& \multicolumn{2}{c|}{VQE ($35\%$)}   \\
    \hline
    % Benchmark & Qubit\# & \multicolumn{2}{c|}{Circuit Metrics} \\
    % \hline
    Compiler $/ $ Metrics & Space--time & Comp time (s) & Space--time & Comp time (s) & Space--time & Comp time (s) \\ \hline
    \myCompilerNameSpace & 43316 & 3083.4 & 44590 & 3186.2 & 4459 & 375.7 \\ \hline
    DASCOT & 62972 & 716.1 & 72709 & 875.3 & 6825 & 17.7\\ \hline
    Liblsqecc & 186795 & 3.5 & 203700 & 3.7 & 45045 & 0.78 \\ 
    \hline
     $\bm{8 \times 2}$ \textbf{grid} (Ave: $38\%$) & \multicolumn{2}{c|}{QFT ($35\%$)} &\multicolumn{2}{c|}{QPE ($38\%$)}& \multicolumn{2}{c|}{VQE ($40\%$)}   \\
    \hline
    % Benchmark & Qubit\# & \multicolumn{2}{c|}{Circuit Metrics} \\
    % \hline
    Compiler $/ $ Metrics & Space--time & Comp time (s) & Space--time & Comp time (s) & Space--time & Comp time (s) \\ \hline
    \myCompilerNameSpace & 38080 & 3000.5 & 41480 & 3188.5 & 3825 & 328.0 \\ \hline
    DASCOT & 58650 & 771.4 & 67235 & 912.2 & 6375 & 19.1 \\ \hline
    Liblsqecc & 169955 & 3.2 & 185250 & 3.4 & 40755 & 0.70 \\ 
    \hline
     $\bm{16 \times 1}$ \textbf{grid} (Ave: $38\%$) & \multicolumn{2}{c|}{QFT ($34\%$)} &\multicolumn{2}{c|}{QPE ($39\%$)}& \multicolumn{2}{c|}{VQE ($40\%$)}   \\
    \hline
    % Benchmark & Qubit\# & \multicolumn{2}{c|}{Circuit Metrics} \\
    % \hline
    Compiler $/ $ Metrics & Space--time & Comp time (s) & Space--time & Comp time (s) & Space--time & Comp time (s) \\ \hline
    \myCompilerNameSpace & 45540 & 3043.0 & 48213 & 3180.6 & 4455 & 382.5 \\ \hline
    DASCOT & 69102 & 923.8 & 79596 & 921.8 & 7425 & 28.1 \\ \hline
    Liblsqecc & 191940 & 2.7 & 210210 & 2.8 & 45465 & 0.56 \\ 
    \hline
    
    \end{tabular}%
    }
  \label{tab:Arch}%
  \end{adjustbox}
\end{table}%

\begin{table}[t]
  \centering
%   \vspace{-5pt}
  \caption{Compilation results for two benchmarks of various sizes}
\begin{adjustbox}{width=\textwidth}
 {\fontsize{9}{12}\selectfont
    \begin{tabular}{|c|c|c|c|c|c|c|c|c|c|c|c|c|}
    \hline
    Compiler & BV (qubit$\#$) & 20  & 40  & 60  & 80  & 100  & DJ (qubit$\#$) & 20  & 40  & 60 & 80 & 100    \\
    \hline
    % Benchmark & Qubit\# & \multicolumn{2}{c|}{Circuit Metrics} \\
    % \hline
    \multirow{2}{*}{\myCompilerName} & Space--time & 594 & 1950 & 3468 & 4693 & 6174 & Space--time & 1089 & 5070 & 11849 & 16606 & 29988 \\ \cline{2-13}
    & Comp time (s) & 27.5 & 46.9 & 51.8 & 57.6 & 59.9 & Comp time (s) & 66.2 & 149.2 & 234.7 & 280.2 & 469.1 \\ \hline
    \multirow{2}{*}{DASCOT} & Space--time & 3465 & 13845 & 29189 & 44042 & 67032 & Space--time & 6039 & 23790 & 52598 & 86279 & 133182 \\ \cline{2-13}
    & Comp time (s) & 1.19 & 1.21 & 1.20 & 1.20 & 1.25 & Comp time (s) & 1.65 & 2.24 & 2.75 & 2.72 & 26.97 \\ \hline
    \multirow{2}{*}{Liblsqecc (tight)} & Space--time & 1692 & 6270 & 12960 & 20538 & 31668 & Space--time & 3420 & 14454 & 28128 & 45612 & 83460 \\ \cline{2-13}
    & Comp time (s) & 0.06 & 0.16 & 0.31 & 0.45 & 0.66 & Comp time (s) & 0.12 & 0.39 & 0.66 & 1.01 & 1.82 \\ \hline
    \multirow{2}{*}{Liblsqecc (sparse)} & Space--time & 5499 & 21375 & 43605 & 65037 & 98049 & Space--time & 10764 & 46125 & 90117 & 140049 & 239568 \\ \cline{2-13}
    & Comp time (s) & 0.09 & 0.31 & 0.58 & 0.82 & 1.22 & Comp time (s) & 0.19 & 0.73 & 1.28 & 1.86 & 3.32 \\ \hline
     Volume Reduction & Ave: $73\%$ & $65\%$  & $69\%$  & $73\%$  & $77\%$  & $81\%$  & Ave: $64\%$ & $68\%$  & $65\%$  & $58\%$ & $64\%$ & $64\%$    \\
    \hline

    \end{tabular}%
}    
  \label{tab:Scalability}%
  \end{adjustbox}
\end{table}%

\begin{table}[t]
  \centering
%   \vspace{-5pt}
  \caption{Comparison with LaSsynth. `-' indicates that the SAT solver does not finish within 24 hours.}
\begin{adjustbox}{width=\textwidth}
 {\fontsize{9}{12}\selectfont
    \begin{tabular}{|c|c|c|c|c|c|c|c|c|c|c|c|c|}
    \hline
    Compiler & BV (qubit$\#$) & 20  & 30  & 40  & 50  & 60  & Random Clifford (qubit$\#$) & 4  & 6  & 8 & 10 & 12    \\
    \hline
    % Benchmark & Qubit\# & \multicolumn{2}{c|}{Circuit Metrics} \\
    % \hline
    \multirow{2}{*}{\myCompilerName} & Space--time & 594 & 1001 & 1950 & 2550 & 3468 & Space--time & 200 & 315 & 637 & 756 & 756 \\ \cline{2-13}
    & Comp time (s) & 27.5 & 28.9 & 46.9 & 50.6 & 51.8 & Comp time (s) & 24.1 & 32.9 & 40.7 & 50.6 & 85.6 \\ \hline
    \multirow{2}{*}{LaSsynth} & Space--time & 396 & 572 & - & - & - & Space--time & 100 & 140 & - & -  & -  \\ \cline{2-13}
    & Comp time (s) & 153.0 & 431.8 & - & - & - & Comp time (s) & 4.8 & 14.7 & - & - & - \\ \hline
     Volume Reduction &  & -$33\%$  & -$43\%$  & -  & -  & - &  & -$50\%$  & -$56\%$  & - & - & - \\
    \hline

    \end{tabular}%
}    
  \label{tab:SAT}%
    \end{adjustbox}
    % \vspace{-3pt}
\end{table}%

\subsection{Scalability Study} \label{Sec:ScalabilityStudy}

Finally, we evaluate whether topology-aware partitioning keeps the realization subproblems manageable on larger-qubit circuits, with the results summarized in Table~\ref{tab:Scalability}.
BV and DJ algorithms are selected as benchmarks, with qubit counts ranging from 20 to 100. The average space--time volume reductions achieved are $73\%$ for BV and $64\%$ for DJ.
We observe that \myCompilerName's compilation time increases near-linearly with system size on these two benchmark families.
The trend of this growth is illustrated in Figure~\ref{fig:CompTime}.

\begin{wrapfigure}{R}{0.55\textwidth}
    \centering
    \includegraphics[width=0.55\textwidth]{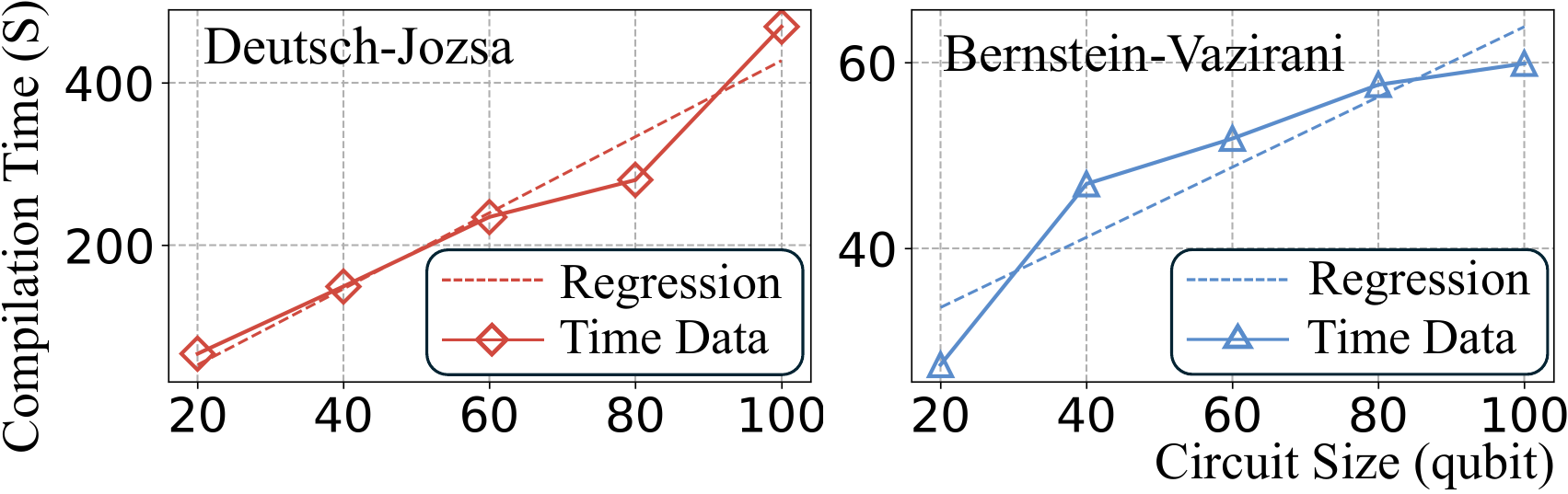}
    \caption{Scalability study: compilation time of \myCompilerNameSpace as circuit size increases.}
    \label{fig:CompTime}
\end{wrapfigure}

We also evaluate another kind of circuit: random circuits with 100 to 500 qubits (Table~\ref{tab:random}), each consisting of two rounds; each round applies Hadamard gates to 20\% of the qubits, chosen at random, followed by CNOTs between fixed pairs of nearby qubits. These circuits have a special structure and mainly test how large a circuit \myCompilerNameSpace can compile. From 100 to 500 qubits, the compilation time of \myCompilerNameSpace grows from 43\,s to 106\,s, more slowly than the circuit size, and its volume is 23--63\% smaller than that of DASCOT, the best baseline, with the gap growing with the circuit size.

\begin{wraptable}{R}{0.575\textwidth}
  \centering
  \caption{Compilation results for random circuits of various sizes}
  \label{tab:random}%
\begin{adjustbox}{width=0.575\textwidth}
 {\fontsize{9}{12}\selectfont   % colour kept inside this group (a bare  leaks out of wraptable)
    \begin{tabular}{|c|c|c|c|c|c|c|}
    \hline
    Compiler & Qubit$\#$ & 100 & 200 & 300 & 400 & 500 \\
    \hline
    \multirow{2}{*}{\myCompilerName} & Space--time & 4410 & 12586 & 18130 & 23534 & 31725 \\ \cline{2-7}
    & Comp time (s) & 43.3 & 45.8 & 47.4 & 65.8 & 106.0 \\ \hline
    \multirow{2}{*}{DASCOT} & Space--time & 5733 & 18879 & 36260 & 60516 & 84600 \\ \cline{2-7}
    & Comp time (s) & 6.65 & 7.79 & 12.38 & 20.81 & 30.18 \\ \hline
    Volume Reduction & Ave: $46\%$ & $23\%$ & $33\%$ & $50\%$ & $61\%$ & $63\%$ \\
    \hline
    \end{tabular}%
}
\end{adjustbox}
\end{wraptable}

 We compare \myCompilerNameSpace with the SAT-solver-based compiler LaSsynth in Table~\ref{tab:SAT}. While the SAT solver can find the optimal solution and achieve results that are approximately $46\%$ better than ours when it succeeds, its scalability is highly limited. As the system size increases, the problem becomes exponentially harder to solve for LaSsynth. We imposed a one-day timeout for all compilation tasks, after which LaSsynth failed to produce results beyond 40 qubits for the BV algorithm with simple structures and 8 qubits for the random Clifford circuit.

This comparison highlights the intended role of \myCompilerName: it does not replace exact synthesis on small Clifford instances, but provides a scalable topology-centric compiler for larger and more general circuits. Because the SAT solver searches the design space exhaustively, LaSsynth finds more compact layouts. Figure~\ref{fig:lassynth} places the two layouts of a 4-qubit circuit of 12 CNOT gates side by side (volume 200 for \myCompilerNameSpace and 100 for LaSsynth) to show where they differ; such comparisons on small instances can guide future improvements of \myCompilerName.

\section{Related Work} \label{sec:Related Work}

\begin{wrapfigure}{R}{0.4\textwidth}
    \centering
    \includegraphics[width=0.4\textwidth]{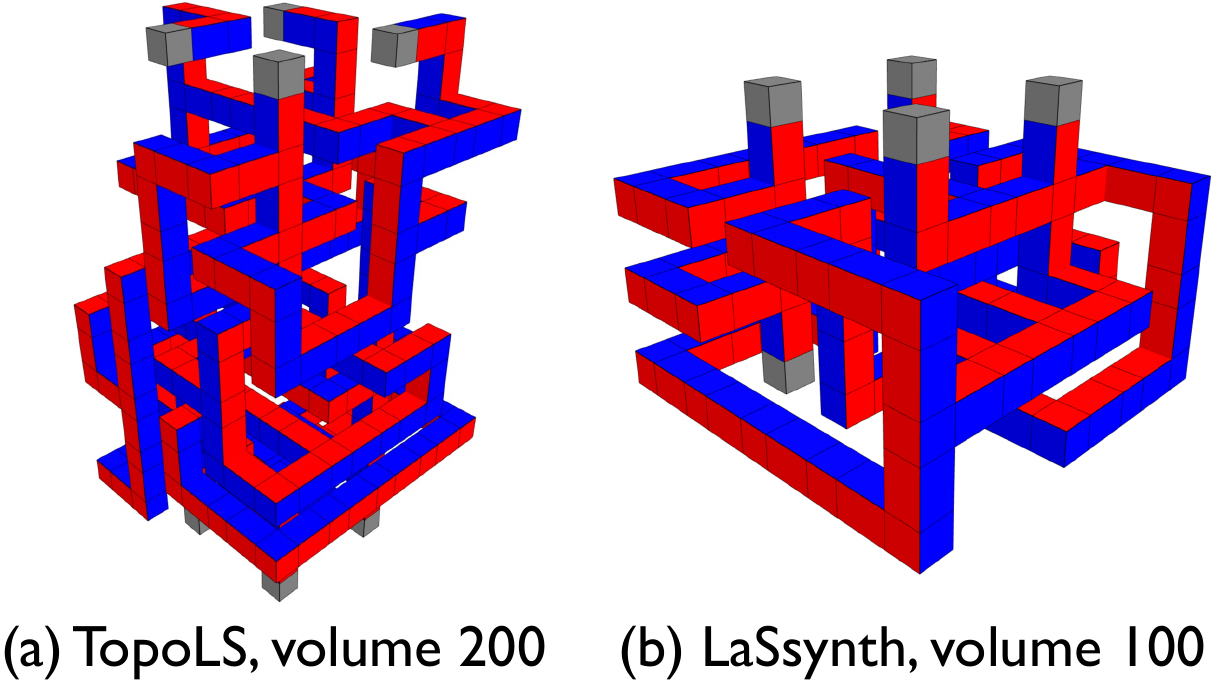}
    \caption{Pipe diagrams of a 4-qubit circuit of 12 CNOT gates; grey cubes are input and output ports.}
    \label{fig:lassynth}
\end{wrapfigure}

\textbf{Gate-Centric Quantum Compilers}: 
Quantum compilation has made significant progress in the last decade, and several industry/academia compilers have been developed~\cite{javadi2024quantum,younis2021berkeley,sivarajah2020t}. Most of them either target the NISQ architectures or compile to the Clifford+T gate set, staying at the purely logical level when targeting fault-tolerant quantum architecture.
Recently, \cite{watkins2024high} and \cite{molavi2025dependency} were proposed to compile and optimize the lattice surgery operations on the surface code. However, they are still designed by the gate-centric philosophy, where the pipe diagrams are constructed by adding new structures for each two-qubit gate individually.
The optimization opportunity from the topological properties of lattice surgery is largely overlooked. 
In contrast, TopoLS is designed to be topology-centric, enabling it to exploit the topological optimization opportunities of lattice surgery when constructing pipe diagrams.

\textbf{SAT Solver for Lattice Surgery}:
In addition to heuristic compilation, Tan et al.~\cite{tan2024sat} developed LaSsynth, an SAT-based approach to synthesize the pipe diagram for a quantum program by encoding the pipe information and constraints into a SAT solver.
Although it can find the pipe diagram with optimal space--time volume for a Clifford circuit block, it suffers from the scalability issue due to the nature of the SAT formulation. 
In contrast, this work can compile and optimize large circuits and support end-to-end compilation for general quantum circuits.

\textbf{ZX-calculus.}
ZX-calculus has been widely used to study quantum computations, including circuit simplification~\cite{duncan2020graph,fischbach2025review,kissinger2019reducing}, circuit reasoning~\cite{van2025optimal,thanos2024automated}, circuit verification~\cite{peham2023equivalence,hurwitz2024simulation}, and the design of measurement-based quantum computation~\cite{backens2021there,kissinger2019universal}. 
Recent work has also explored the connection between ZX diagrams and lattice surgery constructions~\cite{de2020zx,gidney2024inplace,tan2024sat}. However, prior work mainly uses ZX-calculus for reasoning about or verifying lattice surgery protocols, rather than as a compiler intermediate representation. Hand-designed layouts such as magic-state factories~\cite{gidney2019efficient} and the RSA-2048 resource estimate~\cite{gidney2021factor} are assembled manually from small lattice-surgery blocks.
In contrast, \myCompilerName{} employs ZX diagrams as the core IR to support program transformations and the construction of optimized lattice surgery layouts.

\section{Conclusion} \label{sec:Conclusion}

We presented \myCompilerName, a topology-centric compiler that uses ZX diagrams as an intermediate representation for lattice surgery programs. By combining ZX-level transformations, MCTS-based 3D embedding, and topology-aware IR partitioning, \myCompilerNameSpace constructs compact pipe diagrams while maintaining scalability. Experimental results show consistent space--time volume reductions across benchmarks and architectures, with near-linear compilation-time growth on the evaluated large circuits. The approach is not tied to the surface code: code surgery generalizes lattice surgery to quantum LDPC codes and also implements logical operations by measurements~\cite{cross2024improved,williamson2024low,zhou2026genecs}, so ZX diagrams can describe it as well. Extending \myCompilerNameSpace to code surgery would require code-specific operation mappings, validity constraints, and realization procedures.

% Statements required by POPL after the conclusion and before the references; they do not count towards
% the page limit. Order (user, 2026-10-21): "Data Availability Statement", then "Disclosure of Use of AI in Research".

\section*{Data Availability Statement}
\myCompilerNameSpace is available on GitHub at \url{https://github.com/tqec/TopoLS}. In addition, an archived version is provided on Zenodo~\cite{zhou2026topolsartifact}. Both of them include detailed instructions and scripts that facilitate the reproducibility of the experiments presented in the paper.

\section*{Disclosure of Use of AI in Research}
The algorithms and the compiler of \myCompilerNameSpace were designed and implemented by the authors. The authors used an AI coding assistant (Claude Code, Anthropic) to help debug the compiler, and to write and run the scripts for the search-sensitivity experiments (Table~\ref{tab:sensitivity}) and the scalability experiments on random circuits (Table~\ref{tab:random}) and collect their results. The authors reviewed all code changes and scripts produced with its assistance.

\section*{Acknowledgements}
We thank the anonymous reviewers for their insightful feedback. We thank Adrien Suau, Austin Fowler, and Jose A Bolanos for their efforts in integrating our tool into the TQEC ecosystem. We also thank Craig Gidney for insightful discussions on magic state injection.
This work was supported in part by the U.S. Department of Energy, Office of Science, Office of Advanced Scientific Computing Research under Contract No. DE-AC05-00OR22725 through the Accelerated Research in Quantum Computing Program MACH-Q project, the U.S. Department of Energy Early Career Research Award, the U.S. National Science Foundation CAREER Award No. CCF-2338773, and ExpandQISE Award No. OSI 2427020. GL was also supported by the Intel Rising Star Award.

%%
%% The next two lines define the bibliography style to be used, and
%% the bibliography file.
\bibliographystyle{ACM-Reference-Format}
\bibliography{refs}

\newpage
\appendix

\section{Formal Semantics of ZX-calculus} \label{Sec:FormalZX}

ZX-calculus is a graphical language for representing and reasoning about quantum computations~\cite{coecke2011interacting}. 
A ZX diagram is interpreted as a linear map between finite-dimensional Hilbert spaces. 
Given a diagram \(G\) with \(n\) input wires and \(m\) output wires, its denotational semantics is a linear map
\[
    \llbracket G \rrbracket :
    (\mathbb{C}^2)^{\otimes n}
    \rightarrow
    (\mathbb{C}^2)^{\otimes m}.
\]
This interpretation is defined compositionally from the generators of the language. 
Sequential composition of diagrams corresponds to matrix multiplication, and parallel composition corresponds to the tensor product:
\[
    \llbracket G_2 \circ G_1 \rrbracket
    =
    \llbracket G_2 \rrbracket \llbracket G_1 \rrbracket,
    \qquad
    \llbracket G_1 \otimes G_2 \rrbracket
    =
    \llbracket G_1 \rrbracket \otimes \llbracket G_2 \rrbracket .
\]
The basic generators are \(Z\)-spiders, \(X\)-spiders, and Hadamard gates. 
A \(Z\)-spider with phase \(\alpha\), \(p\) input wires, and \(q\) output wires denotes
\[
    \llbracket Z^{\alpha}_{p,q} \rrbracket
    =
    |0\rangle^{\otimes q}\langle 0|^{\otimes p}
    +
    e^{i\alpha}
    |1\rangle^{\otimes q}\langle 1|^{\otimes p}.
\]
Similarly, an \(X\)-spider is defined in the \(X\) basis:
\[
    \llbracket X^{\alpha}_{p,q} \rrbracket
    =
    |+\rangle^{\otimes q}\langle +|^{\otimes p}
    +
    e^{i\alpha}
    |-\rangle^{\otimes q}\langle -|^{\otimes p},
\]
where
\[
    |+\rangle = \frac{|0\rangle+|1\rangle}{\sqrt{2}},
    \qquad
    |-\rangle = \frac{|0\rangle-|1\rangle}{\sqrt{2}}.
\]
Equivalently, an \(X\)-spider can be obtained from a \(Z\)-spider by conjugating each incident wire with a Hadamard gate. 
The Hadamard generator denotes the usual single-qubit unitary
\[
    H =
    \frac{1}{\sqrt{2}}
    \begin{pmatrix}
    1 & 1 \\
    1 & -1
    \end{pmatrix}.
\]
The semantics of a complete ZX diagram is obtained by assigning these tensors to all generators and contracting along connected internal wires.

\section{ZX-to-Pipe Correspondence for Common Quantum Gates} \label{Sec:Gate}

\begin{figure}[H]
    \centering
    \includegraphics[width=0.8\linewidth]{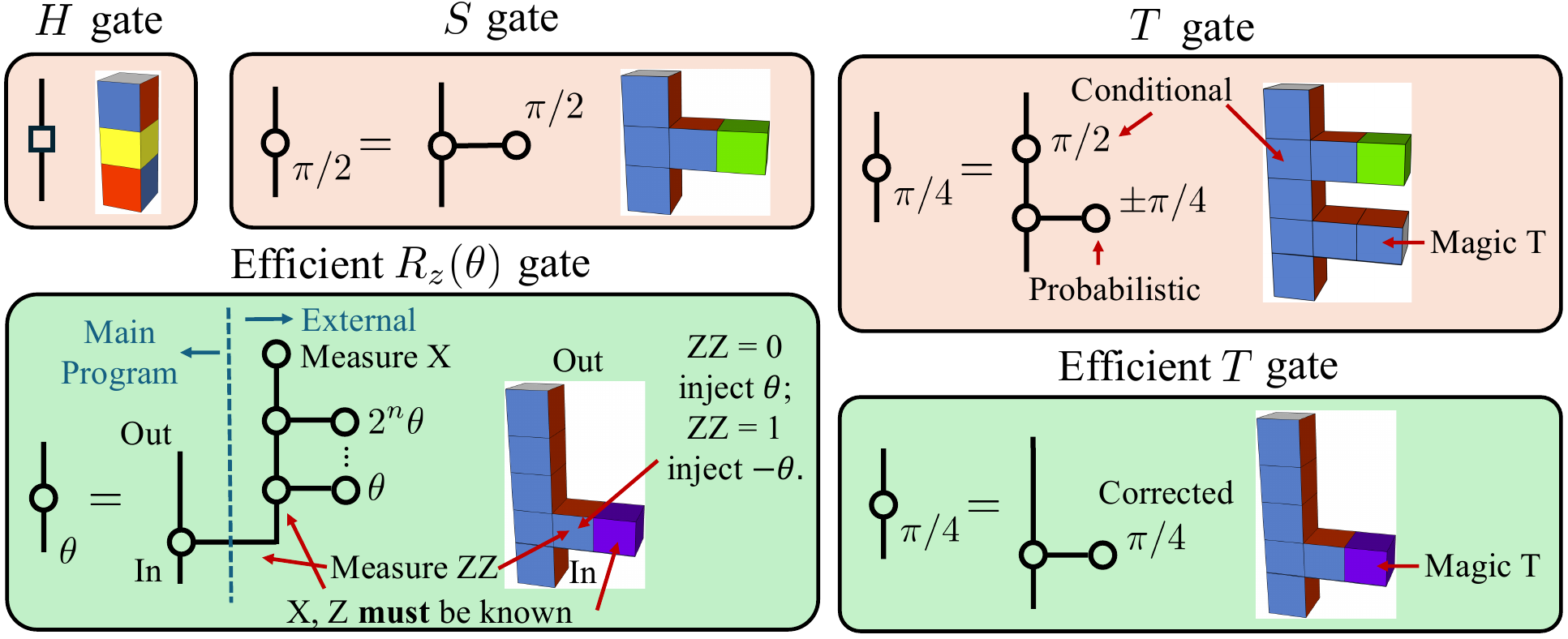}
    \caption{Physical explanation of pipe diagram primitives}
    \label{fig:Gate_}
\end{figure}

We illustrate how commonly used quantum gates are represented in ZX-calculus and realized as pipe structures, as shown in Figure~\ref{fig:Gate_}.

The \textbf{H gate} is represented as a yellow box that swaps the red and blue faces in the pipe diagram. In ZX-calculus, the \textbf{S gate} is represented by a normal Z node connected with a Z node with $\pi/2$ phase. In the pipe diagram, it corresponds to a three-port blue junction together with a phase injection, shown as a green cube in the pipe diagram. The \textbf{T gate} follows a similar interpretation: it involves a three-port blue junction for magic state injection, which produces a $-\pi/4$ phase with $50\%$ probability and necessitates a conditional S correction to obtain the desired $\pi/4$ phase. Recently, the \textbf{$\bm{R_z}$ gate} has gained attention~\cite{sethi2025rescq}, also realized via magic state injection; here, the injection produces $-\theta$ with probability 1/2, necessitating a corrective rotation of $2\theta$ by another injection. This process is probabilistic in depth, continuing until the sum of applied injection angles equals $\theta$, with the failure probability decreasing exponentially with the number of injections.

One limitation of performing magic state injection directly on the logical qubit is that the computation must wait for the injection and the subsequent correction to complete, which can introduce significant latency. 
To mitigate this overhead, we reserve a dedicated magic-state injection port in the pipe diagram (shown as the purple cube in the lower sections of Figure~\ref{fig:Gate_}). Instead of performing the injection directly on the logical qubit, this port can be connected to an external logical-qubit patch, where the magic-state injection and the subsequent correction are performed. This separates the topological connection to the main computation from the resource-producing process outside the main layout.
This design allows the injection process to be carried out externally while the original logical qubit continues executing subsequent operations. From the program's perspective, the in-circuit injection therefore appears to take only one time step, while the actual magic-state preparation and injection are performed asynchronously outside the main execution path. The only requirement is that all $X$ and $Z$ operations tracked in software for the qubit \textbf{must} be clearly known before performing the magic state injection; otherwise the measurement outcome must be interpreted with the wrong Pauli-frame information. This is the causal condition enforced by the injection port.

We give the correctness of our asynchronous $\bm{R_z}$ injection in the following proposition:

\begin{proposition}[Correctness of Efficient $\bm{R_z}$ Injection] \label{prop:Efficient Rz Injection}
If the Pauli frame of the logical qubit is known before the magic-state injection step, then the external-port implementation realizes $R_z(\theta)$ on the logical qubit up to Pauli byproducts that can be tracked in software.
\end{proposition}

\begin{proof}
Suppose that all byproduct $X$ and $Z$ operators have been corrected prior to the $\bm{R_z}$ injection. Then the logical qubit can be assumed to be in the state $\ket{\psi} = \alpha \ket{0} + \beta \ket{1}$, while the ancillary logical patch is initialized in $\ket{+}=\frac{1}{\sqrt{2}}\ket0+\frac{1}{\sqrt{2}}\ket1$.
 
\textbf{Step 1: Joint $ZZ$ Measurement.} There are two possible outcomes, 0 and 1. 
\textbf{Case 1: $ZZ=0$.} The post-measurement state is $\alpha \ket{00} + \beta \ket{11}$. 
\textbf{Case 2: $ZZ=1$.} The post-measurement state is $\alpha \ket{01} + \beta \ket{10}$.

\textbf{Step 2: Magic State Injection.} 
For \textbf{Case 1}, we inject $\theta$ into the ancillary logical qubit, resulting in the state $\alpha \ket{00} + e^{i \theta}\beta \ket{11}$. 
For \textbf{Case 2}, we inject $-\theta$ into the ancillary logical qubit, resulting in the state $e^{-i \theta}\alpha \ket{01} + \beta \ket{10}$, which is equivalent to $\alpha \ket{01} + e^{i \theta}\beta \ket{10}$. The injection process is probabilistic, but after the standard adaptive correction sequence, the desired injection angle is obtained with exponentially small failure probability in the number of attempts.

\textbf{Step 3: Measure Ancillary Logical Qubit.} 
We now measure the ancillary logical qubit in the $X$ basis. 
For \textbf{Case 1}, if the measurement outcome is $\ket{+}$, the resulting state is $\alpha \ket{0} + e^{i \theta}\beta \ket{1}$ as desired; if the outcome is $\ket{-}$, the resulting state is $\alpha \ket{0} - e^{i \theta}\beta \ket{1}$, which can be viewed as a $Z$ byproduct on the desired state. 
For \textbf{Case 2}, if the measurement outcome is $\ket{+}$, the resulting state is $\alpha \ket{0} + e^{i \theta}\beta \ket{1}$ as desired; if the outcome is $\ket{-}$, the resulting state is $-\alpha \ket{0} + e^{i \theta}\beta \ket{1}$, which again can be viewed as a $Z$ byproduct on the desired state.

In fact, the byproduct $X$ and $Z$ operators can be deferred and corrected after the $ZZ$ measurement, but before the magic state injection. For example, consider a state:
\[
\ket{\psi} = X(\alpha \ket{0} + \beta \ket{1}).
\]
If the presence of the $X$ byproduct is known, we can conditionally inject an angle $-\theta$, resulting in the state:
\[
\ket{\psi} = e^{-i\theta}\alpha \ket{1} + \beta \ket{0}.
\]
This differs from the desired state:
\[
\ket{\psi} = \alpha \ket{0} + e^{i\theta}\beta \ket{1},
\]
by an $X$ byproduct operator. The remaining cases can be verified similarly.
    
This approach allows us to avoid waiting before the $ZZ$ measurement, shifting the required correction step to just before the magic state injection instead.
\end{proof}

\section{Proofs}

\subsection{Proof of Proposition~\ref{prop:Spider Fusion}} \label{subsec:proof-spider-fusion}

\textbf{[Semantic Preservation of Spider Fusion]}
Let $G \Rightarrow_{fuse} G_{fuse}$ denote the application of the spider fusion rule in ZX-calculus. 
Then
\[
\llbracket G \rrbracket = \llbracket G_{fuse} \rrbracket .
\]

\begin{proof}
The spider fusion rule is a sound rewrite rule of ZX-calculus that preserves the linear map semantics of ZX diagrams~\cite{coecke2011interacting}. Therefore, applying spider fusion does not change the represented computation.
\end{proof}

\subsection{Proof of Theorem~\ref{theorem:BFS Layer Slicing}} \label{subsec:proof-bfs-layer-slicing}

\textbf{[Validity of BFS Layer Slicing]}
Let $L(v)$ be the BFS distance of a ZX spider $v$ from the input boundary wires, so that the first spider on each wire has $L(v)=1$. Then $L$ is slicing-valid: every ZX edge is either internal to one layer or crosses between two adjacent layers. Every spider $v$ with $L(v)>1$ has a neighbor $u$ with $L(u)=L(v)-1$. Moreover, BFS slicing only annotates the nodes of $G$ with layer labels; it does not add, delete, or rewrite any ZX node or edge, so the annotated IR has the same denotation as $G$.

\begin{proof}
For any edge $(u,v)\in E$, BFS distances from the same source set differ by at most one; otherwise the larger-distance endpoint would have been discovered through the smaller-distance endpoint at an earlier layer. Therefore $|L(u)-L(v)|\leq 1$ for every edge, so the produced layer order satisfies the slicing-validity constraint and gives the backend a well-defined inner/input/output edge classification. This proves the validity of the order annotation. If $L(v)>1$, then $v$ was discovered from a vertex $u$ of the previous BFS layer, which is a neighbor with $L(u)=L(v)-1$. Separately, the BFS pass only records the labels $L(v)$ on existing vertices and leaves $V$, $E$, generator kinds, phases, and boundaries unchanged. Hence the annotated IR represents the same ZX graph with extra compiler metadata, and its denotation is unchanged.
\end{proof}

\subsection{Proof of Proposition~\ref{prop:Semantic Preservation of ZX-to-Pipe Instantiation}} \label{subsec:proof-zx-to-pipe-construction}

\textbf{[Semantic Preservation of ZX-to-Pipe Construction]}
Let
\[
P = \mathcal{I}(G)
\]
be the pipe diagram constructed from the ZX diagram $G$ using the construction procedure after the placeholder defusions described above.
Then
\[
\llbracket \mathcal{T}_P(P) \rrbracket
=
\llbracket G \rrbracket .
\]

\begin{proof}
Each ZX node in $G$ is instantiated as a cube that corresponds to the same ZX node under the interpretation mapping $\mathcal{T}_P$.
The routed connections introduced during embedding connect these cubes by sequences of pipes and may also introduce additional blue or red routing cubes. These additional cubes are incident to exactly two pipes and therefore correspond, under $\mathcal{T}_P$, to degree-two zero-phase spiders. Such spiders are identities in ZX-calculus. Hence routing may change the syntactic ZX graph recovered from the pipe diagram, but only by inserting identity spiders along existing edges.

When placeholder ports are introduced for unresolved outputs, each affected spider is defused into two same-color spiders connected by an identity edge, and the unresolved future edges are attached to the placeholder spider. This preserves the original spider arity across the two connected spiders and is exactly the inverse of a sound spider-fusion step.
Therefore, interpreting the resulting pipe diagram yields a ZX diagram
\[
\mathcal{T}_P(P) = G_{\mathrm{defuse}},
\]
which is obtained from $G$ by applying spider defusion and by inserting degree-two zero-phase identity spiders along routed edges.

Since spider fusion and defusion are sound rewrite rules in ZX-calculus, and degree-two zero-phase spiders denote identity wires, they preserve the semantics of ZX diagrams. Hence
\[
\llbracket G_{\mathrm{defuse}} \rrbracket
=
\llbracket G \rrbracket .
\]
Consequently,
\[
\llbracket \mathcal{T}_P(P) \rrbracket
=
\llbracket G \rrbracket .
\]
\end{proof}

\subsection{Proof of Proposition~\ref{prop:MCTS Validity Invariant}} \label{subsec:proof-mcts-validity-invariant}

\textbf{[Validity Invariant of Pipe Construction]}
Every accepted embedding state maintained by \myCompilerNameSpace satisfies the applicable validity conditions (V1)--(V5) of Definition~\ref{def:ValidPipeDiagram} for all cubes and pipe segments already constructed. If the MCTS search cannot extend such a state to a complete valid diagram for a layer, the compiler rebuilds the current block from its input interface with a canonical gate-by-gate fallback construction, whose boundary pipes match the already realized neighboring interfaces. Consequently, every completed layer emitted by the backend is a valid pipe diagram fragment.

\begin{proof}
The root state is valid because it contains only initialized boundary cubes and no routed internal connection that can violate conditions (V1)--(V5). Assume the current state satisfies all applicable validity conditions. During expansion, the next ZX node is instantiated only as one of the allowed cube types in Figure~\ref{fig:Gate}. Candidate placements are accepted only when they preserve occupancy and boundary well-formedness (V1), and candidate routed connections are accepted only when every inserted segment preserves pipe geometry (V2) and the relevant cube condition: blue/red orientation and color matching (V3), yellow color switching (V4), or green/purple arity (V5). Any inserted routing cube is a degree-two blue or red cube satisfying (V3). Therefore every accepted successor again satisfies the applicable validity conditions. The simulation phase is used only to estimate the quality and extendability of such states; it does not make an invalid state correct. If no valid completion is found for a layer, \myCompilerNameSpace rebuilds the current block from the interface at its start, gate by gate, with a canonical circuit-derived construction whose primitives are the standard valid pipe gadgets described in Section~\ref{Sec:PipeDiagramforLatticeSurgery} and whose boundary pipes match the already realized interfaces. This construction always succeeds. At the start of a block every logical qubit has one port, the ports occupy distinct cells, and nothing lies above them (I3). The construction extends each qubit by a vertical column above its port and realizes each two-qubit gate by a route between two columns in a new time step above all cells placed so far; this region is empty, so the construction exists. The columns end in a new interface that again satisfies (I3). Thus the backend either returns an MCTS-found valid layer fragment or a fallback valid block fragment. By induction over layers and blocks, the completed pipe diagram is valid.
\end{proof}

\subsection{Proof of Theorem~\ref{theorem:IR Realization Contract}} \label{subsec:proof-ir-realization-contract}

\textbf{[\myCompilerNameSpace Satisfies the IR Realization Contract]}
Let $C$ be an input quantum circuit and let $P$ be the pipe diagram compiled by \myCompilerName\ from $C$. Define
\[
G_C = \mathcal{T}_C(C), \qquad
G_P = \mathcal{T}_P(P),
\]
using the translations from Definitions~\ref{def:CircuitToIR} and~\ref{def:PipeToIR}.
Then
\[
\llbracket G_C \rrbracket = \llbracket G_P \rrbracket
\qquad\text{and}\qquad
P\in\mathcal{P}_{\mathrm{valid}}.
\]
Equivalently, $P$ realizes $G_C$ in the sense of Definition~\ref{def:PipeRealization}.

\begin{proof}
We prove the theorem by considering the two stages of \myCompilerName.

\textbf{Step 1: ZX-Level Program Transformation.}
During compilation, \myCompilerNameSpace applies spider fusion to obtain an optimized ZX diagram $G_{\text{fuse}}$.
By Proposition~\ref{prop:Spider Fusion},
these transformations preserve semantics, and therefore
\begin{equation}
\llbracket G_C \rrbracket 
=
\llbracket G_{\text{fuse}} \rrbracket.
\label{eq1}
\end{equation}

\textbf{Step 2: ZX-to-Pipe Construction and Optimization.}
Next, the compiler constructs a pipe diagram from $G_{\text{fuse}}$:
\[
P = \mathcal{I}(G_{\text{fuse}})
\]
by instantiating the spiders of $G_{\text{fuse}}$ as pipe-diagram cubes and embedding them into 3D space.  
By Proposition~\ref{prop:MCTS Validity Invariant}, the embedding returned by the construction stage is a valid pipe diagram satisfying conditions (V1)--(V5) of Definition~\ref{def:ValidPipeDiagram}. By Proposition~\ref{prop:Semantic Preservation of ZX-to-Pipe Instantiation},
this construction preserves the ZX semantics:
\begin{equation}
\llbracket \mathcal{T}_P(P) \rrbracket
=
\llbracket G_{\text{fuse}} \rrbracket .
\label{eq2}
\end{equation}

Let $G_P = \mathcal{T}_P(P)$. Combining Equation~(\ref{eq1}) and Equation~(\ref{eq2}) above yields:
\[
\llbracket G_P \rrbracket
=
\llbracket G_{\text{fuse}} \rrbracket
=
\llbracket G_C \rrbracket .
\]
The validity part follows from Proposition~\ref{prop:MCTS Validity Invariant}, which establishes that the pipe-level construction returns a valid pipe diagram. Therefore, $P$ is a valid pipe realization of the source IR graph $G_C$.
\end{proof}

\section{Benchmark Circuits} \label{Sec:Benchmarks}

The benchmarks in Table~\ref{tab:benchmark} come from three sources and are then compiled to the gates $H$, $S$, CNOT, and $Z$-axis rotations $R_z(\theta)$. \textbf{MQT Bench}~\cite{quetschlich2023mqt}: the GHZ state, the W state, Deutsch--Jozsa, QPE (exact), and VQE are taken from this benchmark suite. \textbf{qqq-wisc}: QAOA and QFT are taken from the GitHub repository \texttt{qqq-wisc/quantum-compiler-benchmark-circuits}; QAOA has one layer ($p=1$) on a random 3-regular graph with 16 vertices, with parameters $\beta=\gamma=1$, and QFT is the exact QFT on 16 qubits (no approximation) with the final swaps. \textbf{Generated with Qiskit}~\cite{javadi2024quantum}: \textbf{BV} uses a secret string drawn uniformly at random (Python random seed 0). \textbf{Grover} follows the Qiskit Grover tutorial on 6 qubits, with the marked states 110111 and 001011 (random seeds 0 and 1) and one Grover iteration. \textbf{Random Clifford} circuits on $n$ qubits consist of $3n$ gates, each drawn uniformly from $\{H, S, \mathrm{CNOT}\}$ and applied to uniformly random qubits (random seed 0).

\begin{table}[h]
% \begin{table}[t]
  \centering
%   \vspace{-5pt}
  \caption{Benchmark Information}
      % \resizebox{\columnwidth}{!}{
      {\fontsize{10}{12}\selectfont
    \begin{tabular}{|m{4cm}<{\centering}|c|c|c|}
    \hline
    Benchmark & Qubit\# & Depth & Gate\#  \\
    \hline
    \multirow{8}{*}{Bernstein-Vazirani (BV)} & 16 & 12 & 31 \\ \cline{2-4}
                                    & 20 & 13 & 50 \\ \cline{2-4}
                                    & 30 & 21 & 78 \\ \cline{2-4}
                                    & 40 & 25 & 102 \\ \cline{2-4}
                                    & 50 & 28 & 125 \\ \cline{2-4}
                                    & 60 & 35 & 152 \\ \cline{2-4}
                                    & 80 & 42 & 199 \\ \cline{2-4}
                                    & 100 & 52 & 249 \\ \cline{2-4}
    \hline
    \multirow{6}{*}{Deutsch-Jozsa (DJ)} & 16 & 17 & 50 \\ \cline{2-4}
    & 20 & 24 & 87 \\ \cline{2-4}
    & 40 & 43 & 162 \\ \cline{2-4}
    & 60 & 63 & 251 \\ \cline{2-4}
    & 80 & 83 & 341 \\ \cline{2-4}
    & 100 & 103 & 424 \\ 
    \hline
    Grover Search & 6 & 457 & 653\\
    \hline
    Quantum Fourier Transform (QFT) & 16 & 236 & 1102 \\
    \hline
    Quantum Phase Estimation (QPE) & 16 & 325 & 1211 \\
    \hline
    Variational Quantum Eigensolver (VQE) & 16 & 34 & 256 \\
    \hline
    GHZ State & 16 & 16 & 16 \\
    \hline
    W State & 16 & 95 & 232\\
    \hline
    Quantum Approximate Optimization Algorithm (QAOA) & 16 & 34 & 136 \\
    \hline
    \multirow{5}{*}{Random Clifford Circuit} & 4 & 8 & 12 \\ \cline{2-4}
    & 6 & 7 & 18 \\ \cline{2-4}
    & 8 & 9 & 24 \\ \cline{2-4}
    & 10 & 9 & 30 \\ \cline{2-4}
    & 12 & 9 & 36 \\ 
    \hline
    
    \end{tabular}%
    }
  \label{tab:benchmark}%
% \end{table}%
\end{table}

\section{End-to-End Example} \label{Sec:E2E}

Figure~\ref{fig:e2e} traces a 2-qubit circuit, a CNOT followed by an $S$ gate on the control qubit (Figure~\ref{fig:e2e}(a)), through all compilation stages.
\textbf{Frontend} (Definition~\ref{def:CircuitToIR}): the CNOT becomes a $Z$-spider on $q_1$ and an $X$-spider on $q_2$ joined by an edge, and the $S$ gate becomes a $Z$-spider with phase $\pi/2$, drawn with its phase on a one-wire spider as in Figure~\ref{fig:ZX}(c) (Figure~\ref{fig:e2e}(b)).
\textbf{Spider fusion} (Section~\ref{Sec:ZXRepresentationandSpiderFusion}): the two adjacent $Z$-spiders on $q_1$ are fused into one $Z$-spider (Figure~\ref{fig:e2e}(c), left).
\textbf{Partitioning} (Section~\ref{Sec:Scalability}): the circuit is small, so it forms a single block.
\textbf{BFS slicing} (Section~\ref{Sec:OptimizetheInstructionOrderThroughTopologicalConnection}): both spiders are adjacent to the inputs, so they form layer~1 (Figure~\ref{fig:e2e}(c), right).
\textbf{Embedding} (Section~\ref{Sec:MCTSfor3DEmbedding}): MCTS embeds the spiders of layer~1 one at a time. It first places the $Z$-spider above the input port of $q_1$, with its phase realized by a green cube (Figure~\ref{fig:e2e}(d)), and then places the $X$-spider above the input port of $q_2$ and routes the CNOT edge between the two spiders (Figure~\ref{fig:e2e}(e)).
\textbf{Result}: since no layer follows, the open wires are routed to the output ports, which completes the pipe diagram (Figure~\ref{fig:e2e}(f)); no fallback is needed.

\begin{figure}[t]
    \centering
    \includegraphics[width=0.65\textwidth]{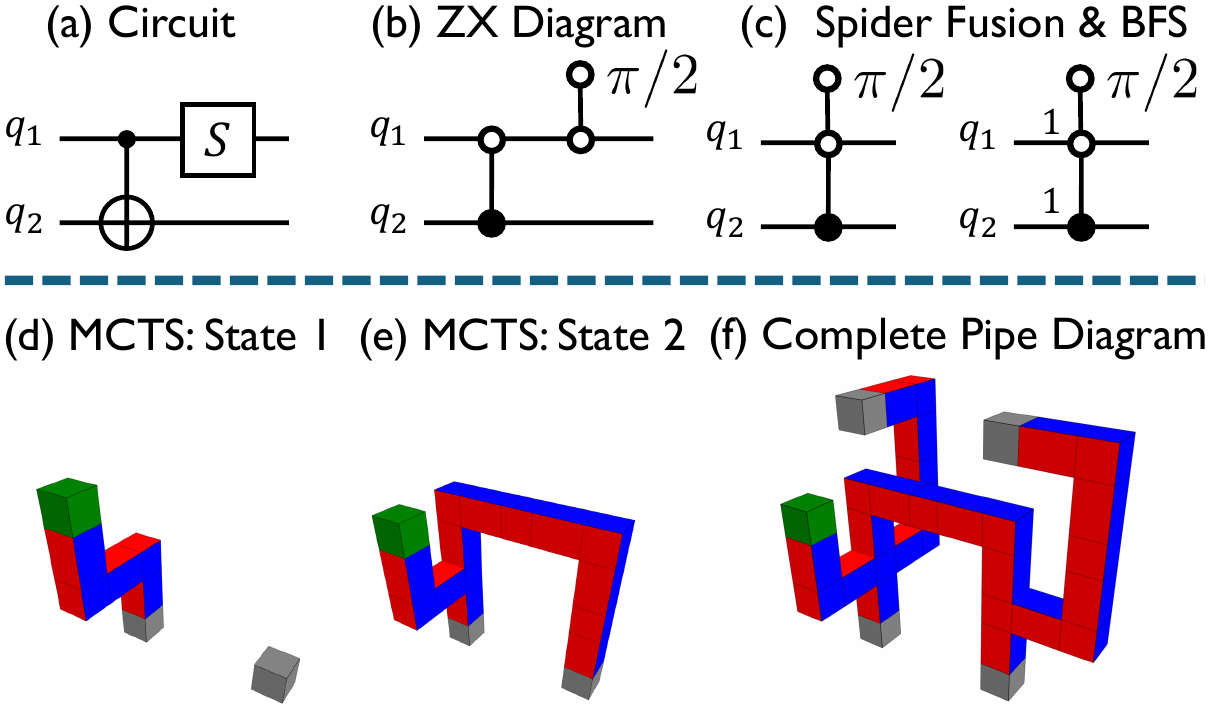}
    \caption{End-to-end compilation of a 2-qubit circuit: (a)~circuit; (b)~ZX diagram; (c)~after spider fusion (left) and BFS slicing (right); (d,e)~MCTS states after embedding the first and the second spider; (f)~complete pipe diagram.}
    \label{fig:e2e}
\end{figure}

\end{document}